\documentclass[11pt]{article}

\usepackage{amsmath,amsfonts,amssymb,amsthm}
\usepackage[margin=1in]{geometry}
\usepackage[colorlinks]{hyperref} 
\usepackage{dsfont}
\usepackage{comment}
\usepackage{bbm}
\usepackage{MnSymbol}
\usepackage{mathdots }
\usepackage[toc,page]{appendix}
\usepackage{algorithm}
\usepackage{algorithmic}
\usepackage{float}
\usepackage{xspace}

\usepackage{caption}
\usepackage{subcaption}

\usepackage{graphicx}
\usepackage{color}
\usepackage{algorithmic}
\usepackage{algorithm}
\usepackage{clrscode}
\usepackage{epsfig, epstopdf}                  
\usepackage{stfloats}
\usepackage{enumerate}

\theoremstyle{theorem}
\newtheorem{theorem}{Theorem}
\newtheorem{lemma}{Lemma}

\newtheorem{corollary}{Corollary}

\theoremstyle{remark}

\theoremstyle{proof}

\def\be{\begin{equation}}
\def\ee{\end{equation}}
\def\ba{\begin{eqnarray}}
\def\ea{\end{eqnarray}}

\def\h{\hskip 1cm }
\def\hh{\hskip 2cm}
\def\la{\langle}
\def\ra{\rangle}
\def\a{\alpha}

\def\m{\mu}

\def\x{{\bf x}}
\def\y{{\bf y}}
\def\z{{\bf z}}
\def\la{\langle}
\def\ra{\rangle}

\newcommand{\abs}[1]{\left\lvert{#1}\right\rvert}
\newcommand{\norm}[1]{\left\lVert{#1}\right\rVert}
\newcommand{\ket}[1]{|{#1}\rangle}

\newcommand{\bra}[1]{\langle{#1}|}

\newcommand{\braket}[2]{\langle{#1}|{#2}\rangle}

\begin{document}

\begin{center}
{\Large \bf Time independent quantum circuits with local interactions}\\
\vspace{1cm} 
Sahand  Seifnashri\footnote{sahand.seifn@gmail.com}, Farzad  Kianvash\footnote{farzadkianvash@gmail.com}, Jahangir Nobakht\footnote{jahangir.n.b@gmail.com}\\  and\\  Vahid  Karimipour\footnote{vahid@sharif.edu}

\vspace{5mm}

\vspace{1cm} Department of Physics, Sharif University of Technology,\\
P.O. Box 11155-9161,\\ Tehran, Iran
\end{center}
\vskip 3cm


	\begin{abstract}
		
	Heisenberg spin chains can act as quantum wires transferring quantum states either perfectly or with high fidelity. Gaussian packets of excitations passing through dual rails can encode the two states of a logical qubit, depending on which rail is empty and which rail is carrying the packet. With extra interactions in one  or between different chains,  one can introduce interaction zones in arrays of such chains, where specific one or two qubit gates act on any qubit which passes through these interaction zones. Therefore, universal quantum computation is made possible in a static way where no external control is needed. This scheme will then pave the way for a scalable way of quantum computation where specific hardware can be connected to make large quantum circuits. Our scheme is an improvement of a recent scheme where we have achieved to borrow an idea  from quantum electrodynamics to replace non-local interactions between spin chains with local interactions mediated by an ancillary chain.  
			
	\end{abstract}
	
	\section{Introduction}



	The circuit model of quantum computation, which is the oldest and the most well studied model of quantum computation, is very similar to the classical model of computation, its basic features are that quantum circuits are drawn as horizontal lines representing flow of qubits on which, instead of classical gates unitary quantum operators are acting. In the same way that the elementary classical gates like AND, OR and NOT can be joined in various ways to implement any Boolean function on $n$ bits, in quantum circuits, a universal set of one and two-qubit unitary gates can be joined in a suitable way to produce any unitary gate to any desired level of  accuracy.  This similarity is seen in any diagram of quantum circuits, like the one shown in figure (\ref{fig:CQCircuit}). However, the similarity stops here and indeed there is a world of difference between the two models. \\

	\begin{figure}[t]
		\centering
		\includegraphics[width=18cm, height=12cm]{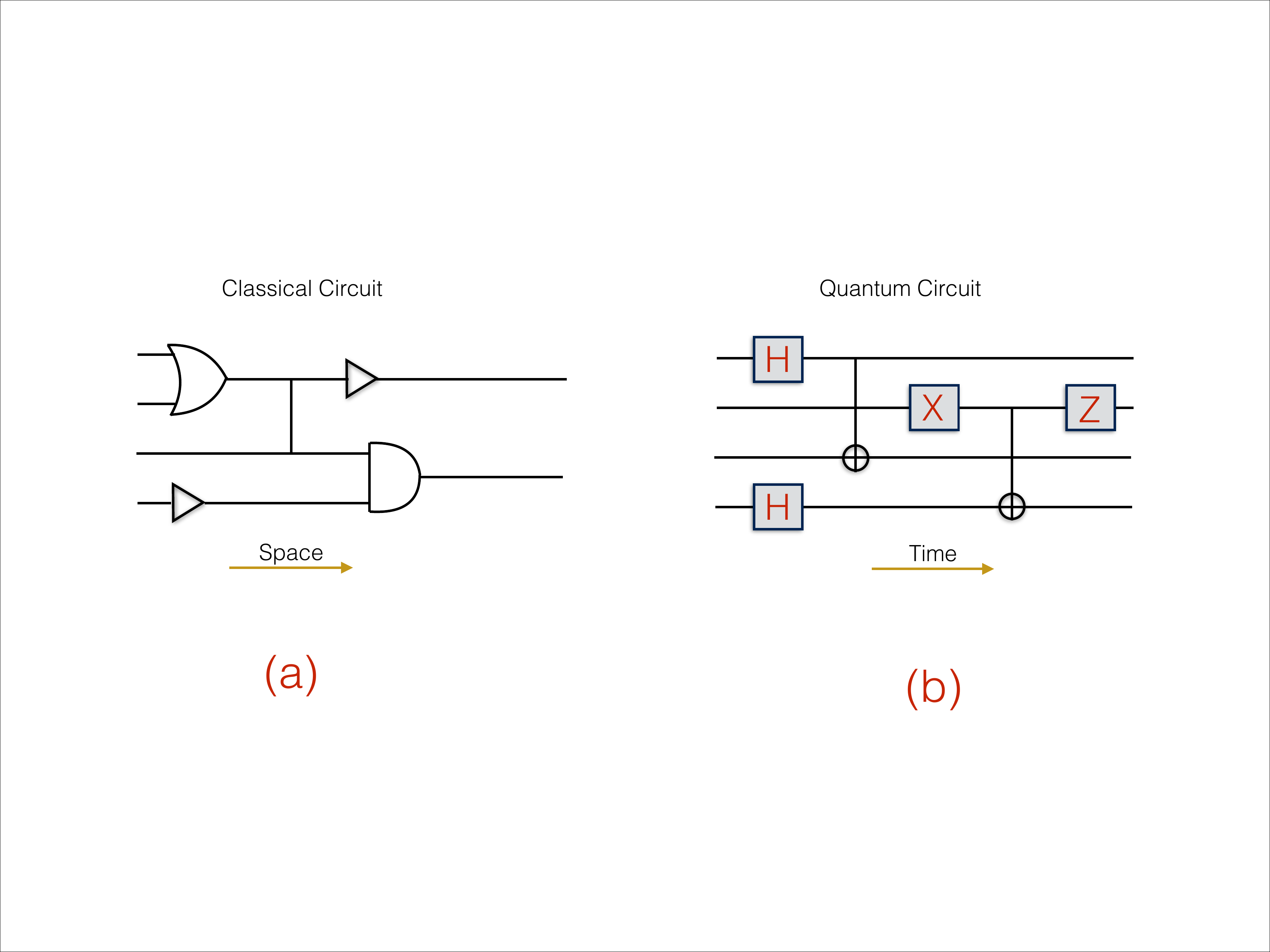}
		\caption{(Color Online) In a classical circuit, horizontal lines corresponds to wires and the gates to localized operations fixed in space acting on bits flowing through wires.  In a quantum circuit horizontal lines correspond to flow of time and gates correspond to localized operations  in time, acting on qubits in their fixed place. }
		\label{fig:CQCircuit}
	\end{figure}
	
	In figure \ref{fig:CQCircuit} (a), horizontal lines actually depict wires, and the whole diagram is a static template which shows how different gates which are well localized and '' fixed-in-space",  act on the logical values of the input bits, once they are fed into the circuit. On the other hand, in figure \ref{fig:CQCircuit} (b) there is no wire  at all and horizontal lines and the gates just display when and on which qubit a unitary gate should be applied "in time". The quantum circuit is not a static hardware template, but is more or less the same as the quantum algorithm itself and a great deal of external control should be applied in time  in order to run this circuit or algorithm. In a classical circuit, bits of information (small electrical currents or bunches of electrons) move down the lines and are acted on automatically by the classical gates on their way, while in  a quantum circuit, qubits of information (electrons, photons, spins, etc) stay in their place, while an external controller, applies different quantum gates on them in a particular order described by the quantum circuit. \\

	Of course there are schemes of quantum computation, e.g. optical realizations,  in which "flying" qubits carry information and are acted on by localized and " fixed-in-space" optical elements. However, the drawbacks of such schemes is the weak interactions between photons in non-linear optical realizations, or the probabilistic nature of gate teleportation in linear optical quantum computing \cite{optics}.  One has also to design mechanisms for converting “flying” qubits  into “stationary” qubits of other types which can interact strongly  \cite{flying} and transfer information over small distances. There are other schemes which are much easier to control and manipulate, like NMR \cite{NMR}, but are generally known not be scalable.\\
	
	 Intensive study of spin chains as quantum wires in the past decade \cite{bose1, ben, bayat, gong, ekert, ksa, woj, ak, woj, ak, jex, kay, ksa, woj, ak, jex, kay, ksa, pst1,pst2,pst3,pst4,pst5} has revealed an interesting synthesis between these two demands. The very appealing property of this solution is that it allows one to 	
	  directly communicate the information, without converting it to another form. Of particular interest are those schemes where information can be routed either passively, with no control on the overall system, or control only over a small portion of the system \cite{dual, Linden, BB, woj, ekert, pst2}, or only global external control over the entire system without any individual addressing \cite{kay, ksa}. \\
	
	Therefore, it is  quite tempting, both theoretically and practically to design schemes where  wires can be arrays of Heisenberg spin chains, and qubits can be excitations which flow down these chains \cite{bose1, ben, bayat, gong, ekert, ksa, woj, ak, jex, kay, ksa}. These chains can be joined to each other and in certain interaction zones, their interaction can be such that when excitations pass through these zones, specific one or two qubit gates act on them, without any external control. The emphasis is here on the absence of "external control" which makes these exactly as classical circuits.  \\

	The most recent attempt in this direction which is the culmination of a long series of investigations \cite{quantum walk, bose, plenio, Twamley, Janzing, Whaley} is reported in \cite{Shor et al.}. In this scheme a pair of periodic Heisenberg spin chains with XY interactions on each of them play the role of a single quantum wire, where a dual rail encoding is used to encode a single qubit, figure (\ref{Dual Rail}). 		
	Single qubit gates are implemented by "rewiring" parts of these chains, in the sense that new types of interactions are imposed on parts of these chains. These will be briefly reviewed in the sequel. As is well-known, universal quantum computation requires in addition to these single qubit gates, the ability to implement also an entangling gate, like the CNOT or CPHASE gate which acts on two qubits \cite{chuang}. In \cite{Shor et al.} this is solved by imposing a non-local interactions between two chains, figure \ref{fig:CPhase} (a).  Certainly this feature is an important hindrance in  experimental realizations of this scheme. \\
	
	This is where our work in this paper is motivated. We want to replace this non-local interaction by a local one and to this end we borrow ideas from quantum electrodynamics, where the long-range interaction between two charged particles is mediated by a photon which interacts locally by each of the charged particles. Therefore, in places where we require a CPhase gate, we add a  third chain between the two chains, figure \ref{fig:CPhase} (b), which mediates the long-range interaction between the two chains by locally interacting with each of them. Note that as shown in \cite{Shor et al.}, the size of the interaction zones in each of the chains and between the chains need to be greater than the width of the Gaussian packets in order not to having appreciable errors. For details of error analysis see \cite{Shor et al.}. The final result is that we now have a set of chains with local interaction which acts like a static quantum circuit. This static quantum circuit is then capable of doing universal quantum computation in the same way that classical circuits can do universal classical computation.\\  
			
	\begin{figure}[t]
		\centering
		\includegraphics[width=14cm, height=10cm]{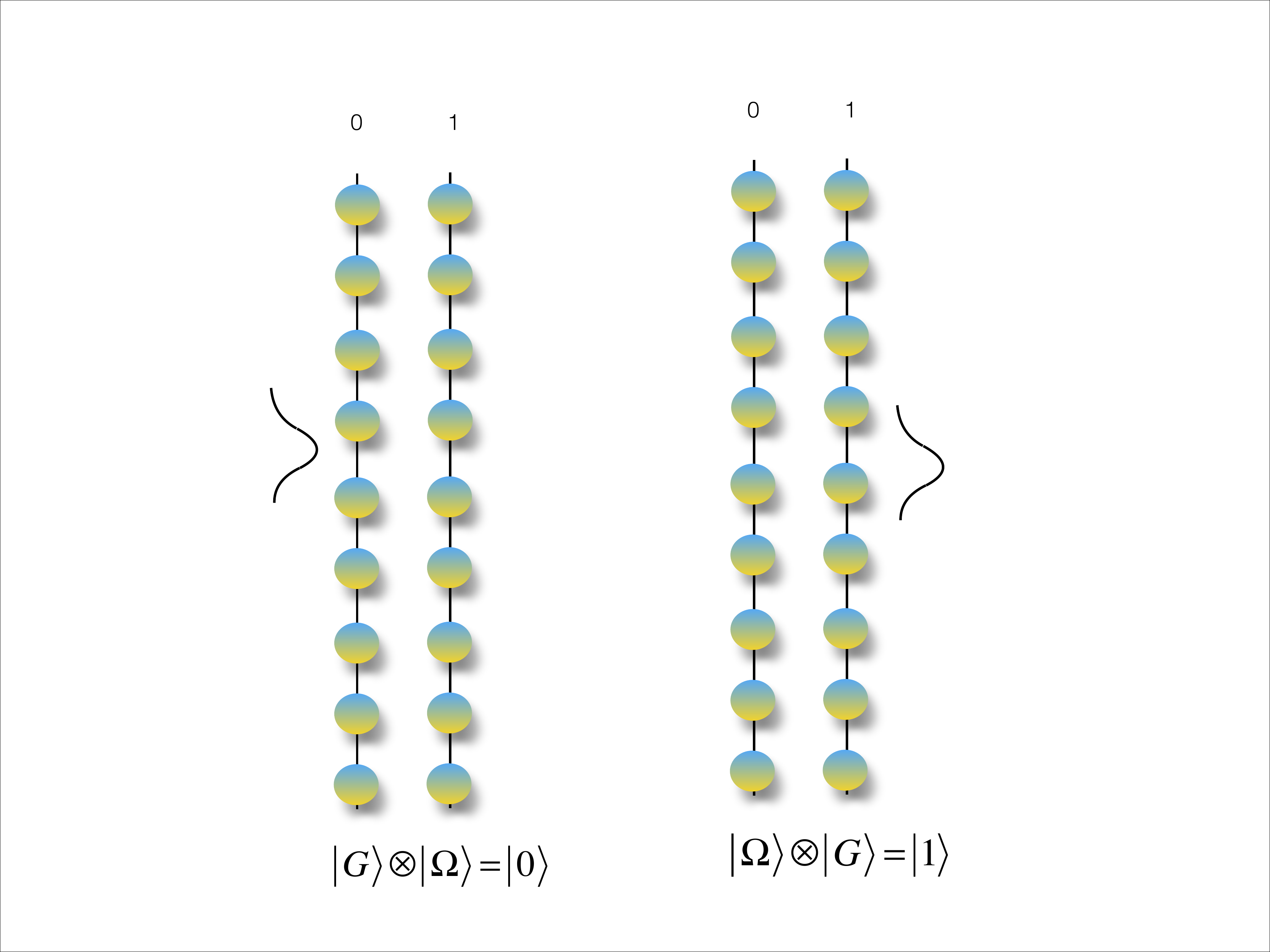}
		\caption{(Color Online) Dual rail encoding for a qubit. The states of the two chains is a logical $|{\bf 0}\ra$ or a logical $|{\bf 1}\ra$, depending on which of the two XY chains is in the vacuum state $|\Omega\ra$ and which one carries a Gaussian packet $|G\ra$.}
		\label{Dual Rail}
	\end{figure}

	The structure of this paper is as follows: In section (\ref{background}), we review in some detail the basic elements of the static quantum circuit of \cite{Shor et al.}, or as they call it Quantum Plinko Machine. In section (\ref{model}), we show how the long-range interactions necessary for the CPhase gate can be implemented by using a third chain which acts as a gauge particle or photon and mediates this long range interaction by local interactions. In fact, we show that a different but still entangling gate can be made in this way. In section (\ref{exp}) we briefly review the experimental progress in realizing spin chains and their manipulation and control in a few systems, like cold atoms and quantum dots and argue that the ingredients needed for realization of  our scheme are at the edge of experimental feasibility within these schemes. Finally, we conclude the paper with a discussion and delegate some of the technical details to the appendices. \\
	
	\section{The static quantum circuit of \cite{Shor et al.}}\label{background}
	For the sake of completeness, in this section we briefly review the work of \cite{Shor et al.}. We do not go into the error analysis carried out in \cite{Shor et al.} and explain only the basic notions: XY chains and their wave packets, the encoding of logical qubits into these chains, the single qubit gates and the CPHASE gate. First we describe our notations.

	\subsection{Notations and conventions} We will be dealing with periodic spin chains consisting of $N$ spins, labeled from $0$ to $N-1$, where the $N$-th site is understood to be the same as site $0$. Pauli operators on the $j$-th spin are denoted by $\x_j, \y_j$ and $\z_j$, where $$\x=\left(
	\begin{array}{cc}
	0 & 1 \\
	1 & 0 \\
	\end{array}
	\right)\ ,\text{ }\y=\left(
	\begin{array}{cc}
	0 & -i \\
	i & 0 \\
	\end{array}
	\right) , \text{ and }\ \z=\left(
	\begin{array}{cc}
	1 & 0 \\
	0 & -1 \\
	\end{array}
	\right)$$ respectively. The Hilbert space of each individual spin is spanned by two states $|\uparrow\ra\equiv |0\ra=\left(\begin{array}{c}1\\ 0 \end{array}\right) $ and $|\downarrow\ra\equiv |1\ra=\left(\begin{array}{c}0\\ 1 \end{array}\right) $. When a chain is in the state of all up spins, i.e.  $|0,0,\cdots 0\ra$, we call it to be in the vacuum state and when a spin in the $j$-th position is down, i.e. $|0,0,\cdots 1, \cdots 0\ra$, we say that there is an excitation or a particle in the $j$-th position. The local operator ${\bf n}_j=\frac{1}{2}(\mathbb{I}-\z_j)=\left(\begin{array}{cc} 0 & 0 \\ 0 & 1\end{array}\right)_j$ detects an excitation in the $j$-th site of the chain and the operator $\hat{N}:=\sum_{j=0}^{N-1}{\bf n}_j$ counts the number of excitations in the whole chain. \\
			
	The configurations of pairs of spin chains, as described in section (\ref{encoding}), will encode logical qubits which will be denoted by bold face numbers inside kets, i.e. $|{\bf 0}\ra$
	and $|{\bf 1}\ra$. 	The Pauli operators on the logical qubits will be denoted by capital letters, 
	\be
		X=|{\bf 0}\ra\la {\bf 1}|+|{\bf 1}\ra\la {\bf 0}|,\hh Y=-i\ |{\bf 0}\ra\la {\bf 1}|+i \ |{\bf 1}\ra\la {\bf 0}|,\hh Z=|{\bf 0}\ra\la {\bf 0}|-|{\bf 1}\ra\la {\bf 1}|.
	\ee

	\subsection{The XY chain}
	Consider a periodic chain of spin 1/2 particles. States in the Hilbert space of a chain are written in the computational basis which have the form $\ket{a_0,a_1,...,a_{N-1}}$ where $a_i={0,1}$.  The Hamiltonian for the chain entails the well-known XY interaction:
	\begin{equation}\label{X-Y}
	H=\frac{1}{2} \sum _{j=0}^{N-1} \x_j \x_{j+1}+\y_{j} \y_{j+1},
	\end{equation}
	where subscript $j$ indicates the $j$-th site of the chain.\\
	This Hamiltonian commutes with $\textbf{S}^z_{total}= \sum _{j=0}^{N-1}\z_j$ so the eigenstates of the Hamiltonian have specific value of $\textbf{S}^z_{total}$. The subspace with $\textbf{S}^z_{total}=N$ is one dimensional which includes the so called vacuum state, 
	
	\be
	|\Omega\ra=\ket{0,0, ..., 0}
	\ee
	and the subspace with $\textbf{S}^z_{total}=N-2$ is called the single excitation subspace. This subspace is N dimensional and is spanned by $\ket{x}$, which is defined as $\ket{0, 0, ...0, 1, 0, ..., 0}$, where 1 is on the $x$-th site. Here we are interested in eigenstates of the Hamiltonian which are in the single excitation subspace, see appendix A. They are defined as
	\begin{equation}
	\ket{\tilde{p}}=\frac{1}{\sqrt{N}}\sum _{x=0}^{N-1}  e^{\frac{2\pi ipx}{N}}\ket{x}, \qquad \text{with energy:} \quad E=2 \cos( \frac{2 \pi p}{N}),
	\end{equation}
	where $p$ is an integer between $0$ to $N-1$.
	A Gaussian wave packet  $\ket{G}$ can be defined as \cite{Dragoman}
	\begin{figure}[t]
		\centering
		\includegraphics[width=14cm, height=10cm]{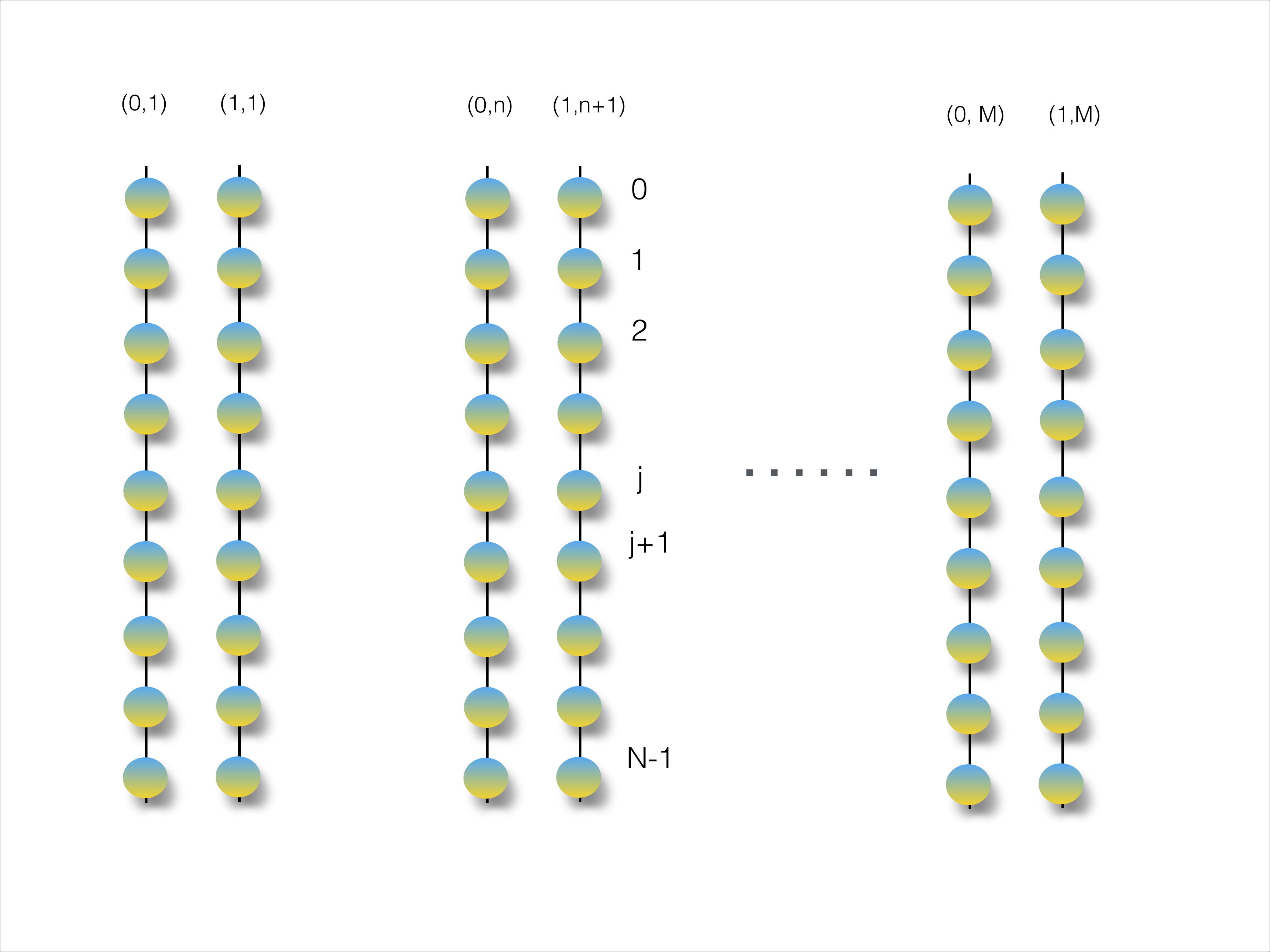}
		\caption{(Color Online) $M$ qubits correspond to $2M$ spin chains. As we will see, for the implementation of entangling gates, it is better to arrange the chains in two layers, the $0$ chains in one layer and the $1$ chains in another parallel layer.  }
		\label{fig:Basic1}
	\end{figure}
	\begin{align} \label{G packet}
		\ket{G}&=\frac{1}{\sqrt{\Delta x\sqrt{\pi }}}\sum _{x=0}^{N-1} \sum _{\alpha =-\infty }^{\infty }e^{\frac{2\pi ip_0x}{N}}  e^{-\frac{\left(\alpha  N+x-x_0\right){}^2}{2 {\Delta x}^2}}\ket{x}\cr
		&=\frac{1}{\sqrt{\Delta p\sqrt{\pi }}}\sum _{p=0}^{N-1} \sum _{\alpha =-\infty }^{\infty }e^{\frac{-2\pi ipx_0}{N}}  e^{-\frac{\left(\alpha  N+p-p_0\right){}^2}{2{\Delta p}^2}}\ket{\tilde{p}},
	\end{align}
	with $2\pi\Delta x \Delta p=N$, where $\Delta x$ and $\Delta p$ are respectively the widths of the packet in the position and momentum spaces.   The group velocity of these packets is given by   
	\begin{equation}
	v_g=\frac{N}{2\pi}\frac{dE}{dp}=-2\sin(\frac{2\pi p}{N}).
	\end{equation}
	By choosing $p=N/4$, the packet will propagate with group velocity $v_g=2$ and with no dispersion, since at $p=N/4$ the derivative of the group velocity (dispersion) is $0$ \cite{Linden, Shor et al.}.
	The unitary evolution of the system $U=e^{-iHt}$ acts as a translation operator on the wave packets which in time $t$ will translate their center from $x_0$ to $x_0+v_gt$ up to a small error \cite{Linden}. \\

	\subsection{The dual rail encoding}\label{encoding}
	In \cite{Shor et al.} the dual rail encoding is used to represent a qubit.  A pair of XY chains represent the $|{\bf 0}\ra$ and the $|{\bf 1}\ra$ states of a single qubit, when they are in the following states, figure (\ref{Dual Rail}): 
	\begin{equation}
		\ket{{\bf 0}} :=  \ket{G} \otimes \ket{\Omega}
	\end{equation}
	\begin{equation}
		\ket{{\bf 1}} :=  \ket{\Omega} \otimes \ket{G}
	\end{equation}
	From eq. (\ref{G packet}) we find that 
	\be
		(\mathbb{I}\otimes \hat{N})\ket{{\bf 0}}=0,\h \text{and}\h (\mathbb{I}\otimes \hat{N})\ket{{\bf 1}}=\ket{{\bf 1}},
	\ee
	Hence, the operator $\mathbb{I}\otimes \hat{N}$ acting on the dual rail detects whether the logical qubit is in the $|{\bf 0}\ra$ or the $|{\bf 1}\ra$ state. We will later use this fact in our local implementation of the entangling gate.  \\ 

	To represent $M$ qubits, $2M$ chains are used. The pair of chains corresponding to the $n$-th qubit are pointed to by the pair of indices $(0,n)$ and $(1,n)$, and all the chains are of equal length $N$, where the spins on each chain are numbers from $0$ to $N-1$, figure (\ref{fig:Basic1}). A Pauli operator like $\x$ which acts on the $j$-th spin of the chain $(a,n)$ ($a=0,1$) is denoted by $\x^{(a,n)}_j$. Thus the Hamiltonian for the free system which has no quantum gate in it is given by

	\begin{equation}
		H_\text{free}= \frac{1}{2} \sum _{j=0}^{N-1}\sum_{a=0,1}\sum_{n=1}^M {\x}^{(a,n)}_j {\x}^{(a,n)}_{j+1}+{\y}^{(a,n)}_j {\y}^{(a,n)}_{j+1}.
	\end{equation}
	Therefore, when there is no extra interaction between the spins, any Gaussian wave-packet flows down the chain without being acted on by any sort of gate. This  acts as a collection of $M$ simple wires of the quantum circuit. 

	\subsection{Single qubit gates}\label{xz gates}
	A single qubit gate acts on a single qubit, or a single wire (a pair of $XY$ chains). Therefore, for simplicity we can drop the extra superscript $n$  from the spin chains  and the operators (pointing to the label of the qubit) and write the free Hamiltonian for the two chains in the form

	\be\label{Hfree}
		H^\text{1-qubit}_\text{free}= \frac{1}{2} \sum _{j=0}^{N-1} \left(\x^0_j \x^0_{j+1}+\y^0_j \y^0_{j+1}+ \x^1_j \x^1_{j+1}+\y^1_j \y^1_{j+1}\right).
	\ee
	Any single qubit unitary can be decomposed, (up to an overall phase) in the form \cite{chuang} 
	\begin{equation}
		U=e^{i\theta_1Z}e^{i\theta_2X}e^{i\theta_3Z}
	\end{equation}
	where $\theta_i$ are real parameters $\in [0,\pi]$. This is described separately for the two types of gates $e^{i\theta Z}$ and $e^{i\theta X}$ in the following.  

	\subsubsection{The $e^{i\theta Z}$ gate}
	To implement this gate, the following interaction is added to the $1-$chain of  (\ref{Hfree}), \cite{Shor et al.}: 

	\begin{equation}
		H_Z=\phi \sum_{j=0}^{N-1} \, \frac{\mathbb{I}^1_j-\z^1_j}{2}
	\end{equation}
	This modified Hamiltonian is still diagonal in the momentum basis but this time with eigenvalues $E\left(\ket{\tilde{p}}\right) = 2 \cos( \frac{2\pi p}{N}) + \phi$ which are shifted by the amount $\phi$. It is now straightforward to see that with adding this extra term, the logical state $\ket{1}$  will gain an extra phase of $\phi$ per unit time, when it passes through the wire, while the logical qubit $|0\ra$ doesn't get any phase. Therefore, the gate $e^{i\phi t Z}$ is implemented on the qubit. \\
 
	\begin{figure}[t]
		\centering
		\includegraphics[width=14cm, height=10cm]{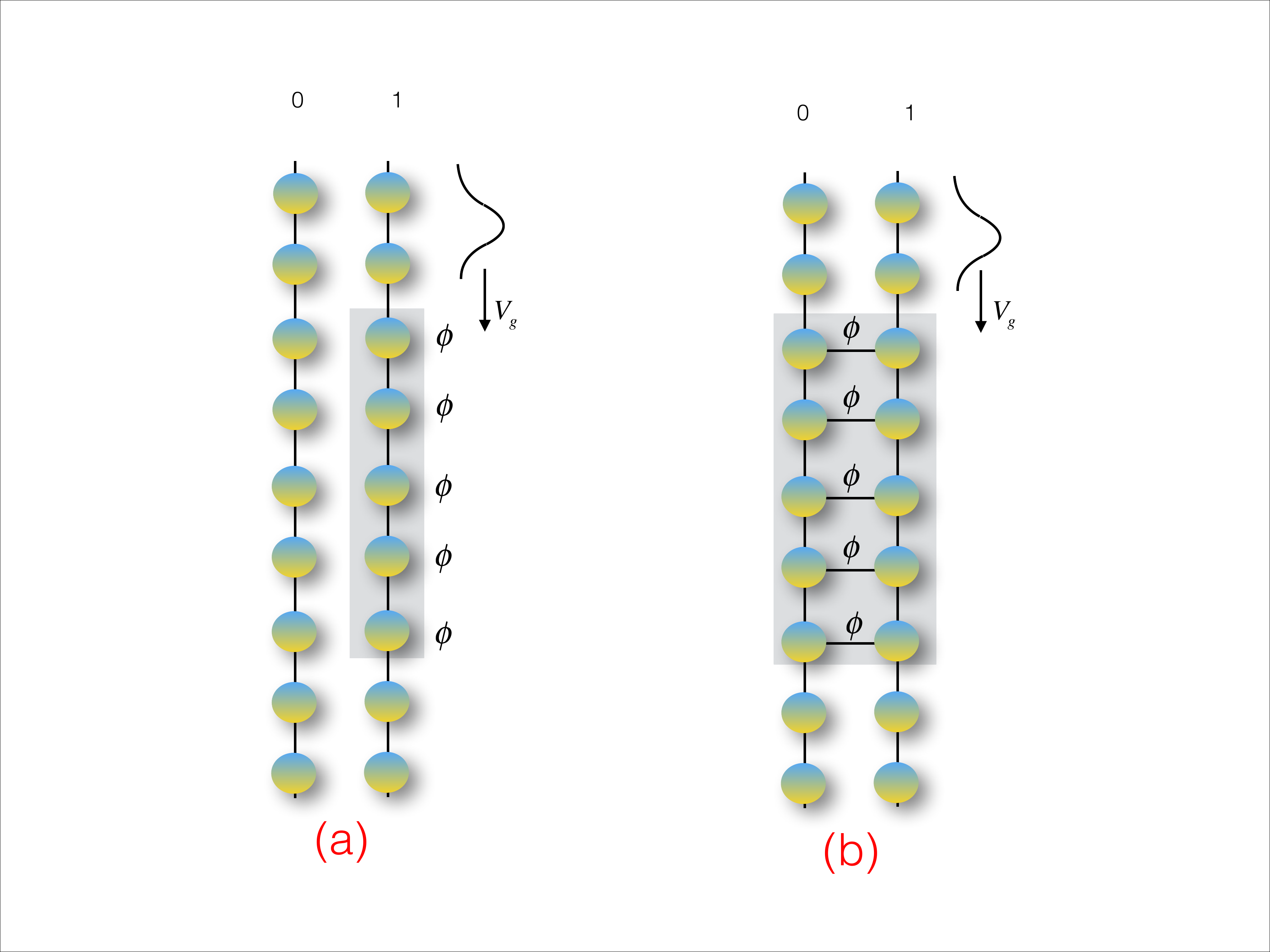}
		\caption{(Color Online) (a) Implementation of a $e^{i\phi Z}$ gate. (b) Implementation of a $e^{i\phi X}$ gate. A curly line indicates the Gaussian wave packet and the boxes shows on which region the gates are acting.  The Gaussian wave packet enters the gate with group velocity $v_g$ and the length of the gate is chosen to be much larger than the width of the packet. The $\phi$ signs depicts the extra interactions for implementations of the gates. }
		\label{fig:X&Zgate2}
	\end{figure}
	
	Notice that in this scheme, each gate (both single and two qubit gates) is implemented in a confined gate block. This means that only a portion of the chains have the added extra interaction terms, see figure (\ref{fig:X&Zgate2}). In fact, when the packets are localized outside of a gate, up to some approximation, they do not feel the presence of these extra terms and we can safely ignore them. Also when they are well localized inside a gate block, again up to some approximation, we can safely assume that the gate spans the whole chain. The detailed error analysis in this regard can be found in \cite{Shor et al.}, we just mention some of their main results here. It has been shown that if we choose the width of the Gaussian packets to be $\Delta x = \Theta\left(N^{1/3}\right)$, and we take the length of each gate block to be $\Theta\left(N^{2/3}\right)$, then we can obtain error $O\left(\frac{1}{(Mg)^{\delta/3}}\right)$ with $N=\Omega\left(M^{3+\delta}g^{3+\delta}\right)$ for any $\delta>0$, where $M$ is the number of qubits, and $g$ is the number of gate blocks.

	\subsubsection{The $e^{i\theta X}$ gate}
	Again we focus on a specific pair of chains pertaining to a single qubit and omit the superscript, the index $n$, pertaining to the label of the qubit. We now add the following XY interaction to the two chains: 
	\begin{equation}
		H_X=\frac{1}{2} \phi  \sum _{j=0}^{N-1} \left(\x_{j}^0 \x_{j}^1+\y_j^0 \y_{j}^1\right),
	\end{equation}
	where the superscripts $0$ or $1$ denotes the two rails pertaining to a single qubit on which the gate is to act. 
	The extra interaction term does not commute with the free Hamiltonian (\ref{Hfree}), however in the one-particle sector it does. The reason is that in the one particle sector, the XY interaction ($h:=\frac{1}{2}(\x\x+\y\y)$) acts just as a simple hopping term, i.e. $h|1,0\ra=|0,1\ra$, $h|0,1\ra=|1,0\ra$ and the order of hopping on the legs and the rungs of the ladder in figure (\ref{fig:X&Zgate2}) enforced by $H_\text{free}$ and $H_X$ is immaterial.  Since $|\tilde{p}\ra$ is the single-particle eigenstate  of the free Hamiltonian (\ref{Hfree}), one can easily see that the following two states 
	\be
		|\tilde{p},\pm\ra:=\frac{|\tilde{p}\ra\otimes |\Omega\ra\pm |\Omega\ra\otimes |\tilde{p}\ra}{\sqrt{2}},
	\ee
	are eigenstates of the new Hamiltonian: 
	\be
		(H_\text{free}+H_X) |\tilde{p},\pm\ra=\left(2\cos(\frac{2\pi p}{N})\pm \phi\right)|\tilde{p},\pm\ra.
	\ee
	Therefore, the Gaussian wave-packets constructed from $|\tilde{p},\pm\ra$, in the form 
	\begin{equation}
		\ket{\pm}= \frac{1}{\sqrt{{\Delta p}\sqrt{\pi }}} \sum _{p=0}^{N-1} \sum _{\alpha =-\infty }^{\infty } e^{-\frac{\left(\alpha  N+p-p_0\right){}^2}{2 {\Delta p}^2}}\ket{ \tilde{p},\pm }
	\end{equation}
	when passing through the pairs of chains in a time $t$, will acquire phases $e^{\pm i\phi t}$. This means that the extra Hamiltonian implements the gate $e^{i\theta X}$ on the single qubit which has this type of interaction box on its way. 

	\subsection{The two qubit CPHASE gate}
	Finally, for  performing a controlled phase operation or CPhase gate, first the four chains of the two qubits  are arranged in such a way that the two $1-$ chains of the two qubits, are adjacent to each other, figure \ref{fig:CPhase} (a), \cite{Shor et al.}. The chains $(1,n)$ and $(1, n+1)$ are adjacent to each other. Actually this arrangement makes half of the $1-$ chains near each other and the half far from each other. So a better arrangement will be to put all the $0$ chains in one layer and all the $1-$ chains in another layer above it. The idea of \cite{Shor et al.} for implementing a CPHASE gate is to impose a long-range interaction between all the spins of the two $1$ chains as shown in figure \ref{fig:CPhase} (a). Again omitting the superscripts corresponding to the labels of the qubits and focusing only on the two $1$-chains, the interaction is of the form  
	\begin{equation}
		H_\text{CPhase}= \phi \sum _{i\leq j=0}^{N-1}\frac{\mathbb{I}_i^{1}-\z_i^{1}}{2}\otimes\frac{\mathbb{I}_j^{1}-\z_j^{1}}{2}=\phi\sum_{i\leq j=0}^{N-1} \left(\begin{array}{cccc} 0 & 0 & 0 & 0 \\ 0 & 0 & 0 & 0 \\ 0 & 0 & 0 & 0 \\ 0 & 0 & 0 & 1\end{array}\right)_{i,j}
	\end{equation}
	This Hamiltonian does not commute with the free Hamiltonian for the chains (\ref{Hfree}), but in the one-particle sector it does. If we order the chains as $0, \ 1, \ 1',\  0'$, then a momentum eigenstate like $|\Omega\ra\otimes |\tilde{p_1}\ra\otimes |\tilde{p_2}\ra\otimes |\Omega\ra$, where two packets  are moving on the chains $1$ and $1'$ has the energy $2 \cos(\frac{2 \pi  p_1}{N})+2 \cos(\frac{2 \pi  p_2}{N})+\phi$, while the other states where there is only one or no packet in the two middle chains have the same energy as before, without an extra $\phi$. Since the four possible logical qubits correspond to the following four states on the chains 
	\ba \label{logical quibit}
		|{\bf 0},{\bf 0}\ra\equiv (|G\ra\otimes |\Omega\ra)\otimes (|G\ra\otimes |\Omega\ra),\cr
		|{\bf 0},{\bf 1}\ra\equiv (|G\ra\otimes |\Omega\ra)\otimes (|\Omega\ra\otimes |G\ra),\cr
		|{\bf 1},{\bf 0}\ra\equiv (|\Omega\ra\otimes |G\ra)\otimes (|G\ra\otimes |\Omega\ra),\cr
		|{\bf 1},{\bf 1}\ra\equiv (|\Omega\ra\otimes |G\ra)\otimes (|\Omega\ra\otimes |G\ra),
	\ea
	and in view of \eqref{logical quibit}, one sees that only when two qubits pass through the chains in a time $t$ simultaneously they pick up a phase $e^{it\phi}$, this interaction implements the two-qubit gate {\it CPHASE}
	\be
		\textit{CPHASE} \: |{\bf a},{\bf b}\ra=e^{i t \phi {\bf a}{\bf b}} \, |{\bf a},{\bf b}\ra.
	\ee
	Now we move on to replace this non-local interaction with a local one and come up with a different entangling gate. 

	\section{Implementation of an entangling  gate with local interaction}\label{model}
	In this section we will introduce a new method for implementing an entangling gate. This gate together with the one qubit gates of section (\ref{xz gates})  will then comprise a universal set of quantum gates. The  scheme of \cite{Shor et al.} for performing a CPHASE gate  on two qubits, required non-local interactions between the $1-$rails in which the two qubits were encoded. Therefore,  we introduce an ancillary rail between the two desired $1-$rails, where the ancillary rail has an XY Hamiltonian in a constant external transverse magnetic field. The necessity of this magnetic field will be explained later on.  This ancillary rail will locally interact with two $1-$rails and effectively produce non-local interactions between the two rails (see figure \ref{fig:CPhase} (b)). The general idea here is similar to that from quantum electrodynamics where the photon generates long-range interactions between charged particles. This analogy will help us to build the model intuitively.\\
	\begin{figure}[t]
		\centering
		\includegraphics[width=14cm, height=10cm]{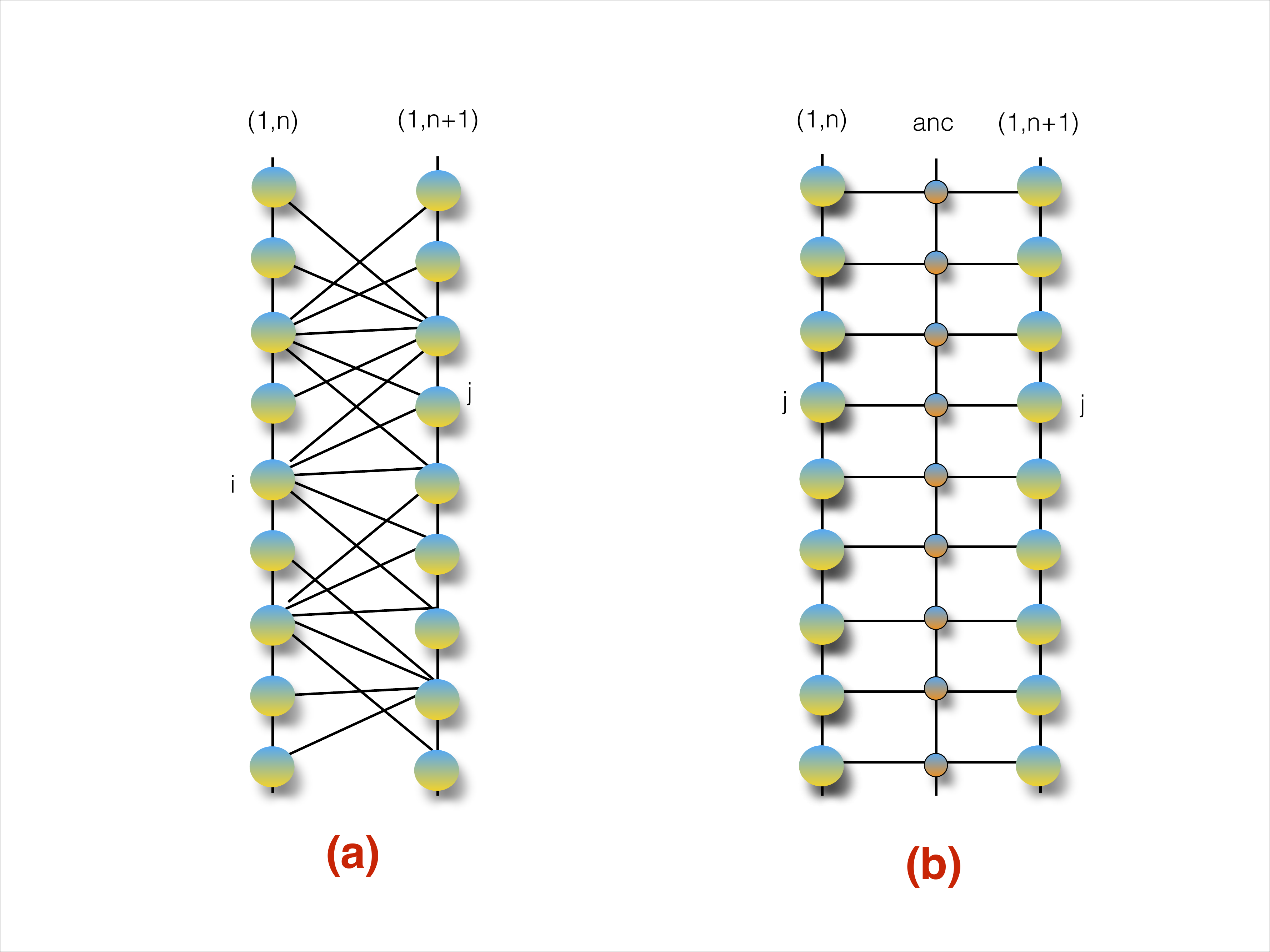}
		\caption{ (Color Online) (a) The CPHASE gate of \cite{Shor et al.} with non-local interactions. (b) The CPHASE gate with local interactions via an ancillary chain. }
		\label{fig:CPhase}
	\end{figure}

To this end we will add an ancillary spin chain in the middle of the two main chains $(n,1)$ and $(n+1, 1)$ and arrange so that it has a doubly degenerate ground state and the energy gap is so large to effectively restrict the dynamics to the degenerate ground space. Denoting the two ground states by $|0\ra_{anc}$ and $|1\ra_{anc}$ and inspired by the basic vertex of quantum electrodynamics, we expect an effective potential as shown in figure (\ref{figphoton}). Working backward from this effective potential we find the Hamiltonian of the ancillary chain to be 

\begin{figure}[t]
		\centering
		\includegraphics[width=14cm, height=10cm]{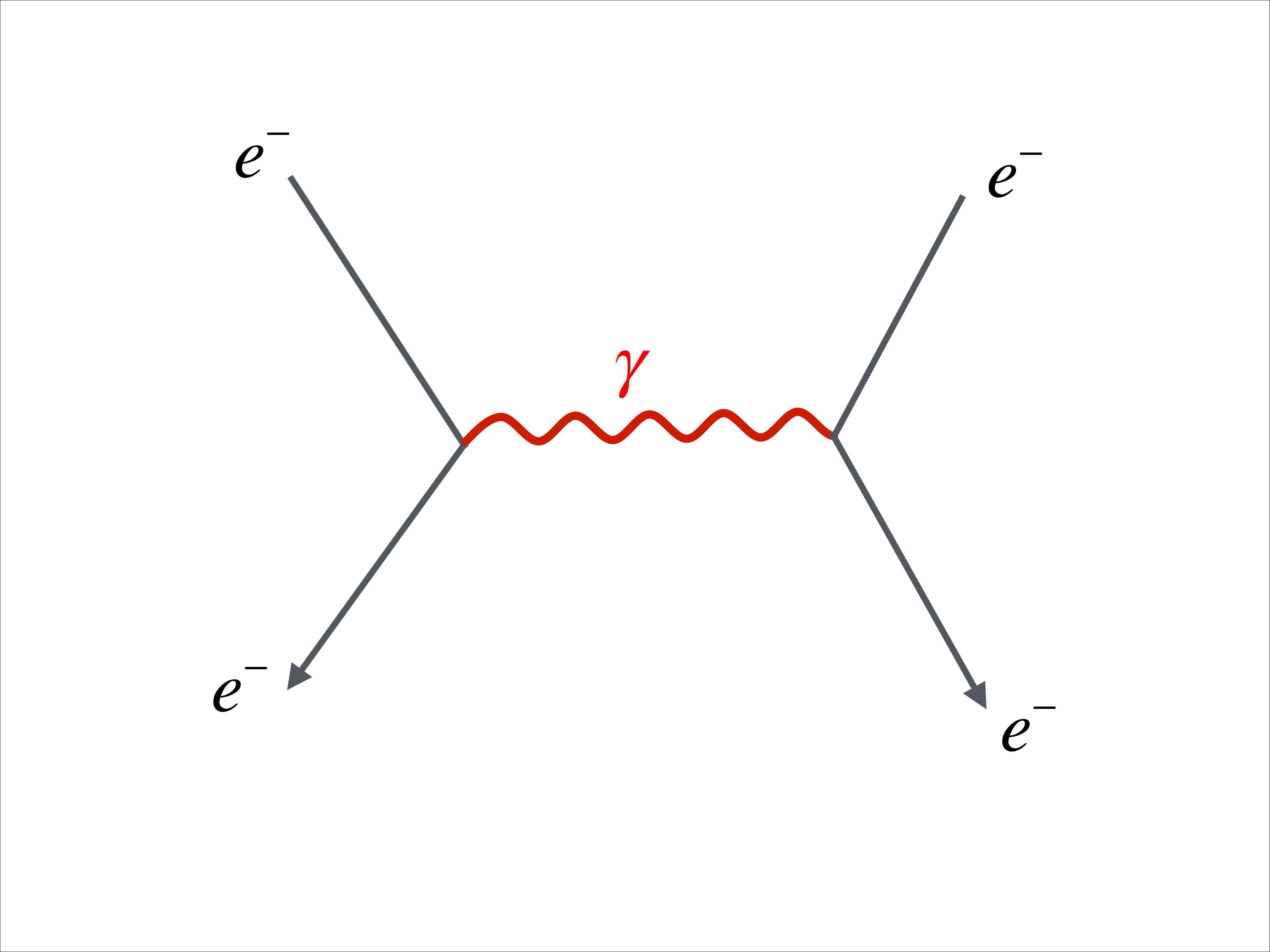}
		\caption{(Color Online) A photon, locally interacting with electrons, mediates long range interactions between them. The ancillary chains in this article for implementing CPHASE gate between the rails carrying qubits is inspired by this effect. }
		\label{fig:qed}
	\end{figure}
	
	\begin{figure}[t]
		\centering
		\includegraphics[width=14cm, height=10cm]{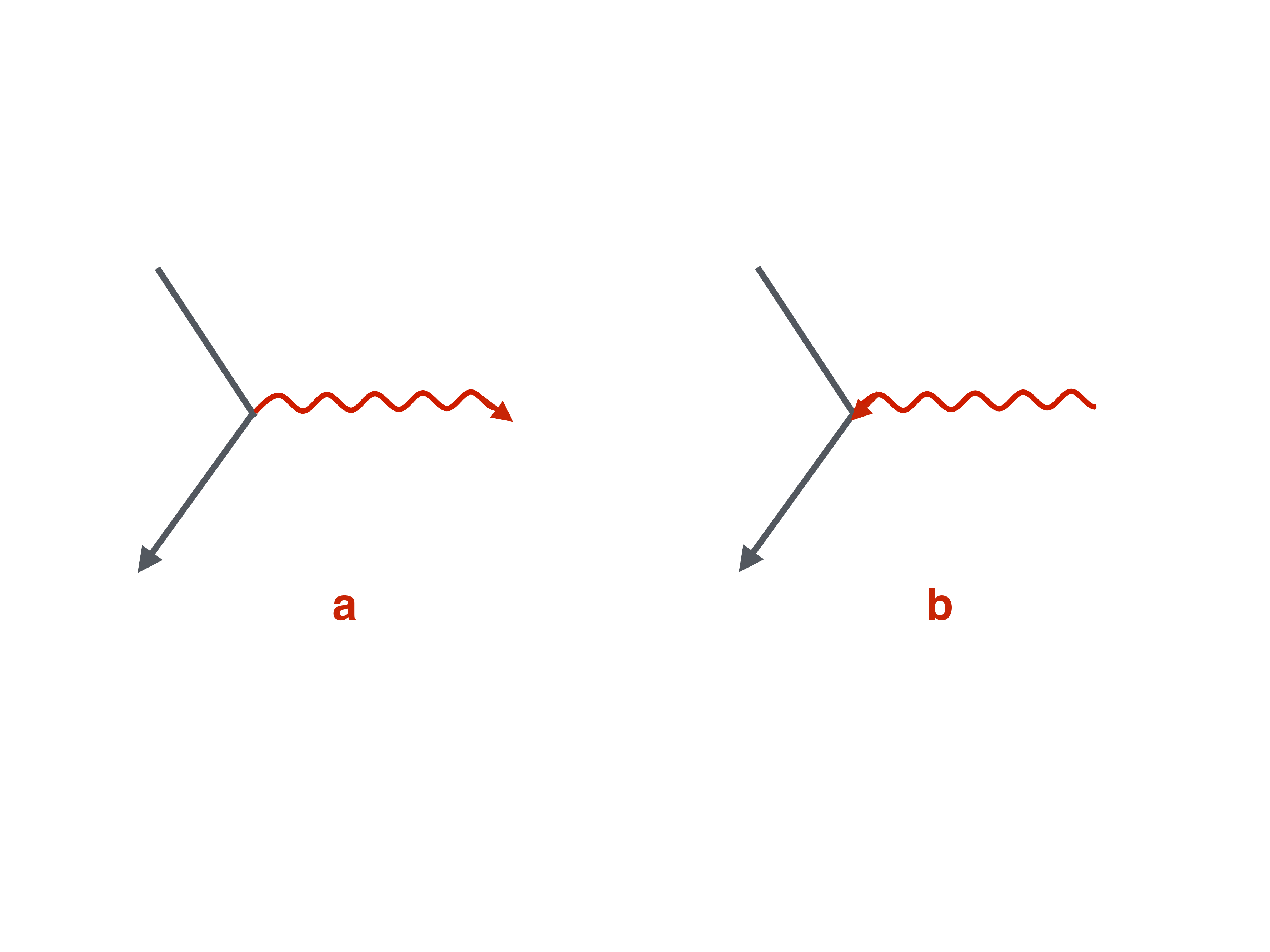}
		\caption{(color online) The basic vertex of QED and the way it inspires the effective interactions (\ref{effint}).  The ancillary qubit plays  the role of the photon with two states $|0\ra_{anc}$ and $|1\ra_{anc}$.  The states of the electron are denoted by  $|0\ra$ and $1\ra$.  The vertex in (a) translate to  $V_j|1\ra\otimes |0\ra_{anc} = |1\ra\otimes |1\ra_{anc}$ and the vertex in (b) to $V_j|1\ra\otimes |1\ra_{anc} = |1\ra\otimes |0\ra_{anc}$, leading to $V_j=n_j\otimes X_{anc}$. 		
		Later in the article and for notational convenience the two states $|0\ra_{anc}$ and $|1\ra_{anc}$ are denoted by $|\Omega\ra$ and $|\Psi\ra$ respectively.  }
		\label{figphoton}
	\end{figure}

	\begin{equation}\label{Hanc}
		H^\text{anc}=\frac{1}{4m} \sum_{i=0}^{N-1} \left(\mathbb{I}-\z_i-\frac{\x_{i}\x_{i+1}+\y_{i}\y_{i+1}}{2}\right),
	\end{equation}	
	where hereafter the superscript 'anc' stands for the ancillary rail. The role of the external magnetic field is to make the ground state degenerate. The reason for the  overall factor of $\frac{-1}{4m}$, will be explained later. The standard way for solving this Hamiltonian is to use the Jordan-Wigner transformation \cite{jordan wigner}. We will use this transformation in  appendix A when we want to find the full spectrum of the model and perform the error analysis in Appendix B. For the present discussion where we are only interested in the ground states, we follow a simpler approach and write (\ref{Hanc}) in the following explicit form:
	\begin{equation}
		H^\text{anc}=\frac{1}{8m} \sum_{j=0}^{N-1}{\bf h}_{j,j+1}
	\end{equation}	
	where
	\be
		{\bf h}_{j,j+1}
		=\left(2\mathbb{I}-\z_j-\z_{j+1}-\x_{j}\x_{j+1}-\y_{j}\y_{j+1}\right)=
		\left(\begin{array}{cccc} 0 & &&\\ & 2 &-2 & \\ & -2& 2 & \\ &&& 4\end{array}\right)_{j,j+1},
	\ee
	which shows that the operator ${\bf h}$ is a positive operator. It is then readily seen that the following two states are eigenstates with zero energy and hence, due to the positivity of ${\bf h}$, the degenerate ground states of $H^\text{anc}$:
	\be \label{logical qubits}
	|\Omega\ra=|0,0,\cdots 0\ra, \h |\Psi\ra:=\frac{1}{\sqrt{N}}\sum_{j=0}^{N-1}|j\ra.
	\ee 
	The argument for $|\Omega\ra$ is simple, because ${\bf h}$  has the product state $|0,0\ra$ as a ground state. The argument for $|\Psi\ra$ is based on the observation that  $\frac{1}{2}\left(\x_i\x_{i+1}+\y_i\y_{i+1}\right)$ acts as a hopping operator, under which $(|0,1\ra\leftrightarrow |1,0\ra)$ and $|00\ra$ and $|11\ra$ are annihilated. The operator ${\bf n}_i=\frac{\mathbb{I}-\z_i}{2}$ also acts as a number operator which gives a total energy equal to zero for $|\Psi\ra$.	\\	
			
	We will later show that the parameter $m$ can be chosen so that there is a sufficiently large gap between the ground space of the ancilla and its excited states. In this ground space we can define the following Pauli operators which act on the two ground states and will be needed in our calculations in the sequel:
	
	\ba \label{logical gates}
	&&\hat{X}^\text{anc}:=|\Omega\ra\la \Psi|+|\Psi\ra\la \Omega|,\cr
	&&\hat{Y}^\text{anc}:=-i\ |\Omega\ra\la \Psi|+i\ |\Psi\ra\la \Omega|,\cr
	&&\hat{Z}^\text{anc}:=|\Omega\ra\la \Omega|- |\Psi\ra\la \Psi|.
	\ea

	We are now ready to introduce the interaction between the ancillary rail and the $1-$rails, which lead to the effective implementation of CPHASE gate.  Since we are concentrating on two specific rails, we use a simpler notation for the Pauli matrices and omit the superscripts.  First consider a block of the ancillary rail and one of the $1-$rails, say the left one, figure \ref{fig:CPhase} (b). We denote  this a $V_\x$ block for reasons to be clear soon. \\
	
	The $j$-th  site of the ancillary rail is connected to its adjacent site (figure \ref{fig:CPhase} (b)),  in the left $1-$rails with couplings  given by 
	\begin{equation}\label{siteV}
		V^{1,\text{anc}}_\x=e\sum_{j=0}^{N-1}\frac{\mathbb{I}-\z_j^1}{2} \otimes \x_j^\text{anc},
	\end{equation}
	where $e$ is a coupling strength. 	  
	A similar interaction exists between the ancillary rail and the right $1-$rail as shown in figure \ref{fig:CPhase} (b). To see what type of interaction this coupling induces on the subspace of logical qubits (\ref{logical qubits}), we note that  by taking the coupling $\frac{1}{4m}$  large enough, and Hence, producing a sufficiently large gap between the ground and the excited states of the ancillary rail, we can restrict the dynamics in the ancillary rail to the two-dimensional ground space spanned by the two vectors $|\Omega\ra$ and $|\Psi\ra$.   In Appendix B, we will show that if we take $\frac{1}{m}> 2N^2$, then restricting the dynamics to the ancillary ground space causes an error which is less than $O(\frac{1}{N^{{1}/{6}}})$. 
	Therefore, we need the effective interaction induced in this subspace which is given by 
	$V^{1,\text{anc}}_\text{x, eff} = (\mathbb{I}\otimes  P_0)V^{1,\text{anc}}_\x(\mathbb{I}\otimes P_0)$, where 
	\be \label{project}
	P_0:=|\Omega\ra\la \Omega|+|\Psi\ra\la \Psi|,
	\ee
	is the projection operator on the ground space of the ancillary rail. In view of the form of $V^{1,\text{anc}}_\x$ and (\ref{project}) we have to calculate the following:
	
	\be
	P_0 \, \x_j P_0=(|\Omega\ra\la \Omega|+|\Psi\ra\la \Psi|)\x_j(|\Omega\ra\la \Omega|+|\Psi\ra\la \Psi|).
	\ee
	We now note that, since $\x_j$ creates an extra particle on an empty site $j$ or remove a particle from this site, then two of the matrix elements vanish, namely $\la \Omega|\x_j|\Omega\ra=\la \Psi|\x_j|\Psi\ra=0$. We also find that $\la \Omega|\x_j|\Psi\ra=\la \Psi|\x_j|\Omega\ra=\frac{1}{\sqrt{N}}$. Therefore, the effective interaction turns out to be of the form  
	
	\be\label{effint}
	V^{1,\text{anc}}_\text{$\x$, eff}=e\frac{1}{\sqrt{N}} \hat{N}^\text{1} \otimes \hat{X}^\text{anc},
	\ee	
	where $\hat{N}^\text{1}$ is the number operator on rail $1$ and $\hat{X}^\text{anc}$ is the Pauli operator $X$ on the ancillary rail (\ref{logical gates}). We remind the reader that the number operator $\hat{N}^\text{1}$ detects the existence of a Gaussian packet on rail $1$.
	In a similar way we can construct another block where the interactions are of the type $\frac{ \mathbb{I}-\z_j}{2}	\otimes \y_j$. This will then lead to the effective interaction
	
	\be
	V^{1,\text{anc}}_\text{$\y$, eff}=e\frac{1}{\sqrt{N}} \hat{N}^\text{1} \otimes \hat{Y}^\text{anc} \label{V}
	\ee	
	
	From these types of interaction between the ancillary rail and the $1-$rails, and by adjusting the signs of the coupling $e$, we construct two blocks of interactions between the ancillary rail and the two $1-$rails adjacent to it and call them respectively simply by $V_X$ and $V_Y$, as in figure (\ref{fig:VVV}). They are defined as follows, where we have simplified the notation, that is,  instead of denoting the two $1-$rails by $(n,1)$ and $(n+1,1)$ we simply denote them by $1$ and $1'$:  
	\ba
	V_X&=&e\frac{1}{\sqrt{N}} \left[\hat{N}^\text{1}-\hat{N}^\text{1'}\right] \otimes \hat{X}^\text{anc}\cr
	V_Y&=&e\frac{1}{\sqrt{N}} \left[\hat{N}^\text{1}-\hat{N}^\text{1'}\right] \otimes \hat{Y}^\text{anc}.
	\label{V}
	\ea	
	With these two blocks, we can now implement an entangling two-qubit gate which make universal computation possible. As we will shortly show the two-qubit gate is given by
	
	\be \label{cphase representatio}
	\Lambda_\phi=\left(\begin{array}{cccc} 1 &&&\\ & e^{-i\phi}&& \\ && e^{-i\phi} & \\ &&&1\end{array}\right),
	\ee
	which is clearly an entangling gate for generic values of $\phi$. Therefore, the effect of three consecutive blocks of the form (\ref{protocol}) as shown in figure (\ref{fig:VVV}), is given by $\Lambda_\phi$:  
	\begin{equation} \label{protocol}
		\Lambda_\phi :=e^{-i\left(\hat{N}^{1}-\hat{N}^{1'}\right)\frac{\pi}{4}\hat{Y}^\text{anc}}
		e^{i\left(\hat{N}^{1}-\hat{N}^{1'}\right)\phi \hat{X}^\text{anc}}
		e^{i\left(\hat{N}^{1}-\hat{N}^{1'}\right)\frac{\pi}{4} \hat{Y}^\text{anc}}
	\end{equation}
	
	\begin{figure}[t]
		\centering
		\includegraphics[width=17cm, height=12cm]{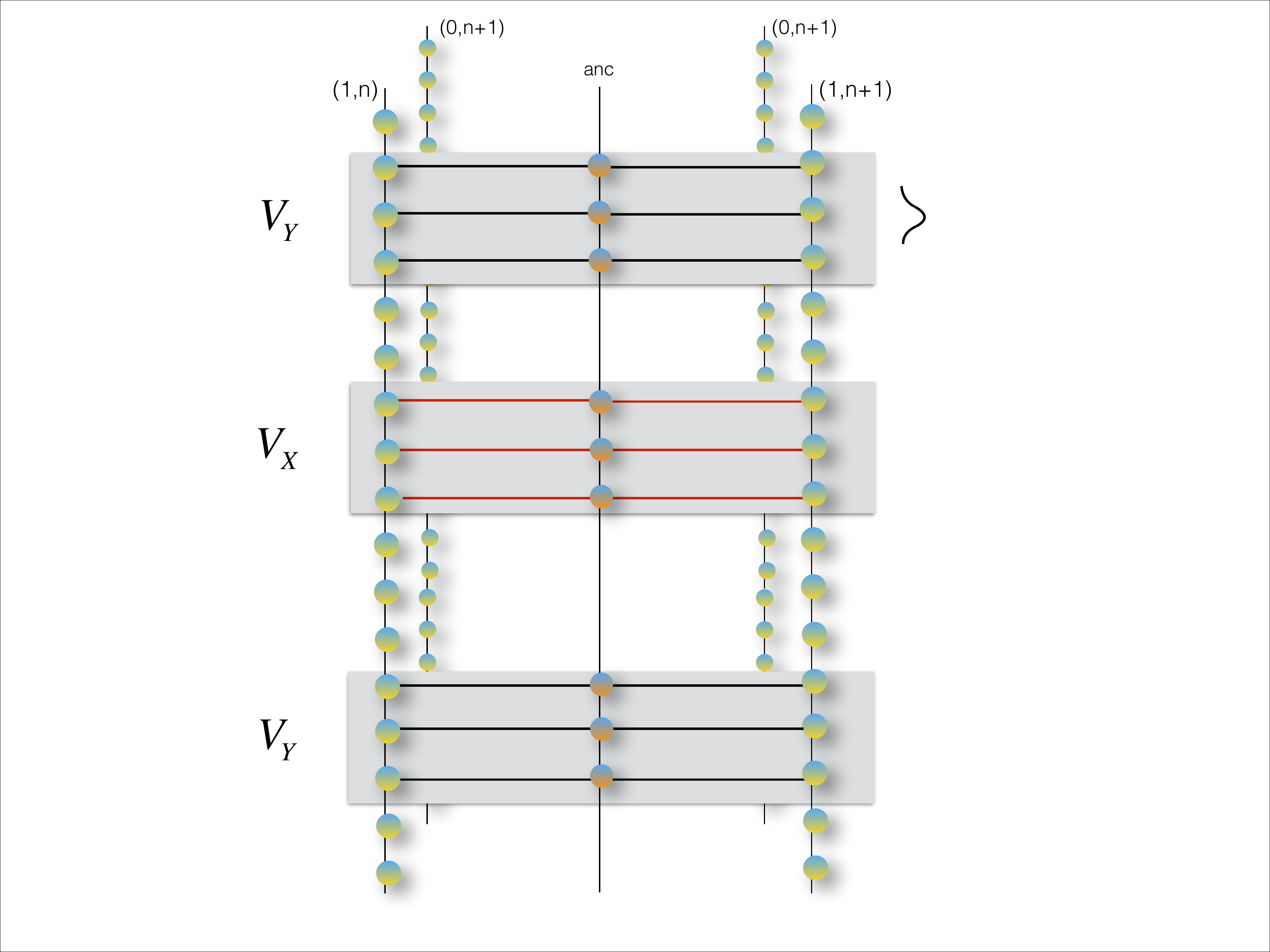}
		\caption{(Color Online) Three successive blocks of the type $V_Y$ and $V_X$, as discussed in the text, implement a controlled phase gate $\Lambda_\phi$, introduced in eqs. (\ref{cphase representatio}) and (\ref{protocol}), between qubits $n$ and $n+1$. We have arranged the $0$ rails and the $1$ rails in two separate layers.}
		\label{fig:VVV}
	\end{figure}
	
	Now we have to find the effect of the above operator on the logical states of the two qubits for which the $1-$rails are part of. We also initialize the state of the ancillary rail at $|\Omega\ra$. First from () we note that  when a logical qubit is in the $|{\bf 0}\ra$ state there is no Gaussian packet in its $1-$rail and when it is in the state $|{\bf 1}\ra$, then there is a Gaussian packet in its $1-$rail. To be precise we have to calculate the effect of these states on the following state state, $|\Omega\ra_1\otimes |\Omega\ra_{1'}\otimes |\Omega\ra_{anc},$ and three other states where there is one or two  Gaussian packets on the first two rails. In view of the fact that the existence of a Gaussian packet in a $1-$rail corresponds to a logical state of the dual rail qubit to be $|{\bf 1}\ra$, we use a simplified notation and denote these states simply by $|{\bf 0},{\bf 0}\ra\otimes |\Omega\ra,\ \  |{\bf 0},{\bf 1}\ra\otimes |\Omega\ra$
	and so on, with the understanding that the states of the $0$ rails has been suppressed. We then find that 
	\ba
		\Lambda_\phi|{\bf 0},{\bf 0}\ra\otimes |\Omega\ra&=&|{\bf 0},{\bf 0}\ra\otimes |\Omega\ra\cr
		\Lambda_\phi|{\bf 1},{\bf 1}\ra\otimes |\Omega\ra&=&|{\bf 1},{\bf 1}\ra\otimes |\Omega\ra,
	\ea
	since in both these cases $\hat{N}^1-\hat{N}^{1'}=0$ and all the blocks in (\ref{fig:VVV}) act as identity operators. 
	On the other hand when the states of the two rails are $|{\bf 1}, {\bf 0}\ra$ or $|{\bf 0}, {\bf 1}\ra$, then the operators on the ancilla act as either $e^{-i\frac{\pi}{4}\hat{Y}^\text{anc}}
	e^{i\phi \hat{X}^\text{anc}}
	e^{i\frac{\pi}{4} \hat{Y}^\text{anc}}$ or $e^{i\frac{\pi}{4}\hat{Y}^\text{anc}}
	e^{-i\phi \hat{X}^\text{anc}}
	e^{-i\frac{\pi}{4} \hat{Y}^\text{anc}}$. Both operators are equal to $\left(\begin{array}{cc} e^{-i\phi}& 0 \\ 0 & e^{i\phi}\end{array}\right)$ where we have used the representation of the last operator in the ground space of the ancillary rail spanned by $|\Omega\ra \equiv \left(\begin{array}{c} 1 \\ 0 \end{array}\right)$  and $|\Psi\ra\equiv \left(\begin{array}{c} 0 \\ 1 \end{array}\right)$. When acting on the state $|\Omega\ra$ they both  
	produce only a phase $e^{-i\phi}$. Note that in all cases the  state of the ancillary rail returns to its initial value $|\Omega\ra$. Putting everything together and suppressing the state of the ancilla,  we have proved that  
	
	\ba
	\Lambda_\phi|{\bf 0}, {\bf 0}\ra &=& |{\bf 0}, {\bf 0}\ra\cr
	\Lambda_\phi|{\bf 0}, {\bf 1}\ra &=& e^{-i\phi}|{\bf 0}, {\bf 1}\ra\cr
	\Lambda_\phi|{\bf 1}, {\bf 0}\ra &=& e^{-i\phi}|{\bf 1}, {\bf 0}\ra\cr
	\Lambda_\phi|{\bf 1}, {\bf 1}\ra &=& |{\bf 1}, {\bf 1}\ra,
	\ea
	which shows that $\Lambda_\phi$ is of the form \eqref{cphase representatio}. In this way the universal set of quantum gates is implemented with local interactions between a set of XY chains.
	
	We should point out that although in all the calculations the length of interaction boxes, both for single qubit gates and for the CPHASE gate have been taken to be equal to $N$, this is not necessary.  In fact, as stressed and elaborated in \cite{Shor et al.}, the length of such boxes should only be well larger than the width of a Gaussian packet.
	More specifically, it has been shown that for having the least possible amount of error, the length of such boxes should be equal to $\Theta\left(N^{2/3}\right)$. Now considering eqs. \eqref{V} and \eqref{protocol}, for implementing a phase $\Phi$ in our entangling gate, we should have $\Phi = \frac{e}{\sqrt{N}}t$, where $t = \frac{\text{gate length}}{2}$ is the time that takes for a wave packet to pass the gate block. Therefore, for having a phase $\Phi \approx 1$, we should choose $\frac{\text{gate length}}{2} \cdot \frac{e}{\sqrt{N}} = \Phi \approx 1$, and consequently $\frac{e}{\sqrt{N}} = \Theta\left(\frac{1}{N^{2/3}}\right)$. We will use this fact later in the Appendix B.

	\section{Possible experimental realizations}\label{exp}
	In this section we discuss possibilities of experimental realization of this scheme.  The discussion is only meant to  show that in view of the long series of attempts for experimental realizations of quantum information processing on spin chains, the ideas presented in this paper are not too far from realizations in the future. Therefore, we  draw the attention of the readers to previous proposals  in which the ingredients of the scheme presented here have been realized in one way or another in closely related systems. To this end, we try to answer the following three questions: \\
	\begin{itemize}
	\item{} Do specific systems exist whose interaction can be modeled by a spin chain Hamiltonian with controllable interactions? \item{} Can these systems be initialized to specific quantum states necessary for quantum information processing? \item{} Can such systems be scaled up to include a large number of individual two state (spin) systems?  
	\end{itemize}
	
	Of course it is understood that each of the proposed systems cannot solve all the problems at once and only in the course of time, a specific, possibly hybrid proposal may be developed with optimal solutions of all the problems. Below we list possible partial answers to the above questions.  \\
	\begin{itemize}
	\item{}	
	  First we note that spin chains with controllable couplings with effective interactions of Heisenberg or XY chains have been realized in several systems. A possible setup is  quantum dot arrays \cite{rev, gatearrays},  where the exchange interactions between the trapped electrons in neighboring dots can be modeled as a spin interactions and the couplings can be tuned  by controlling the voltage barriers between neighboring dots. Another setup is  cold atoms in optical lattices \cite{greiner, sherson, ddl}, where it has been shown \cite{ddl} that it is possible to induce and control strong interaction between spin states of neighboring atoms by adjusting the intensity, frequency, and polarization of
the trapping light.  More specifically it is shown  \cite{ddl} that "for sufficiently strong periodic potential and low temperatures, the atoms will be confined to the lowest Bloch band \cite{greiner} " and "their effective Hamiltonian is given by an
the well-known Heisenberg model (XXZ model)". Moreover it is argued in \cite{ddl}, that homogeneous magnetic fields and also Ising interactions which are required in our implementation of CPHASE gates, "can be easily  
turned on and off by adjusting the potential $V_\m$" or the intensity of laser light. It is important to note that what is required in our scheme is to apply these uniform interactions and magnetic fields over a long array of spins and not on individual atoms or a small number of them. \\

Another scheme is the coupled wave-guide arrays where by suitably choosing the distance between waveguides,  effective interactions of the Hamiltonian has been tuned. In these systems, experimental perfect state transfer has been reported, in arrays of length 11 \cite{pst5} and 19 \cite{art19}. 
	  Other less controllable systems are small NMR systems \cite{nmr} and NV-centers in diamond \cite{nvlukin}. 	\\  
	  \item{}
	  
	  We next face the problem of initializing the chains to the desired states. For gapful spin chains, generally cooling the system down to below the energy gap is the standard tool for this purpose, however for gapless spin chains, which is the more ubiquitous case for spin 1/2 systems, one can use adiabatic evolution to put the system in the ground state \cite{abolfazl}. To prepare a logical qubit in a dual rail, as needed in our scheme,  we should prepare a chain to be in a Gaussian wave-packet which should be almost dispersion free. The conditions for such a preparation have been explored in detail in a theoretical paper \cite{Linden}. Since such a Gaussian packet is nothing but a twisted W-state \cite{Linden}, which is a linear superposition of states in which only one spin is excited, they should resemble spin waves and Hence, should be close to eigenstates of the chain Hamiltonian.   \\
	  
	  Nevertheless, at present we do not know of any concrete experimental proposal for preparing such Gaussian wave packets in the single-particle sector of spin chains. One recent development which may be relevant  in this connection has recently been reported in \cite{naturelast}, where a chain of N=21 waveguides, whose couplings can be modeled by an $XY$ Hamiltonian, have been used to perform Discrete Fourier Transform (DFT). It has been shown in \cite{naturelast} that an input signal of Gaussian form, prepared "by focusing a beam from a HeNe laser", can be converted into a Gaussian profile along the chain. According to \cite{naturelast}, the applications of their scheme, reaches among other areas, " qubit storage and  realization of perfect discrete lenses for non-uniform input distributions", which in turn "opens the way to many interesting applications in integrated quantum computation".\\

\item{} Finally, we come to the question whether or not these spin chains can be scaled up to large sizes. This is In fact, the basic property of every viable candidate for quantum information processing. Generally every protocol, whether it be ion trap or optical lattice faces this problem. As explained in the answer to the first question, effective interactions between electrons in arrays of quantum dots can be modeled by spin Hamiltonians and the couplings can be controlled by adjusting the potential barriers between the dots. Recent years have seen groundbreaking results in fabricating Si, Si/SiGe and dopant-based quantum dots \cite{rev, gatearrays}.  Critical advances like isolation of single electrons, the observation of spin blockade, single-shot read-out of individual electron spins, novel ways of on-chip multiplexing \cite{gatearrays}, make them promising candidates for realization of quantum spin systems with long spin coherence times, and single site addressability necessary for quantum computation and spintronics, in particular for a protocol like the one discussed in the present paper.

\end{itemize}

	\section{Summary}
	The trend of performing quantum computation on chains of spins, either in the form of quantum Turing machines \cite{benioff}, or quantum ballistic models \cite{feynman} or billiard balls \cite{margus}  goes back to the 80's, well before the rapid upsurge of interest in quantum computation initiated by Shor's factoring algorithm and the demonstration of universal set of quantum gates by Barrenco et al \cite{barenco}.  It was then with the work of Sugato Bose in \cite{bose1} that Heisenberg spin chains were shown to be good carriers of quantum states of spins. Since then, intensive studies  have shown that Heisenberg chains with simple or engineered XY interactions, can act as perfect or almost perfect carries of quantum states over short distances. The arbitrary state of a spin joined to the left end of such a chain is carried with high fidelity to the right end through the natural dynamics of this chain, without any external control and without the need to individually address the spins. It has also been shown that these chains can be initialized to carry a Gaussian wave packet of excitations. Dual rails can encode the two states of a logical qubit depending on which of the two rails are empty and which one carries a packet. A $2n$ array of such dual rails can then act as quantum wires carrying $n$ qubits. On specific areas in these arrays, local interactions can be implemented between the spins of the chain such that when qubits pass through these interaction zones, specific one and two-qubit quantum gates, necessary for universal quantum computation, act on them, again without any external control \cite{Shor et al.}. In this way static templates or quantum circuit hardware, can be constructed, which when joined to each other can create large scale quantum circuits. The scheme of \cite{Shor et al.} however requires  long range interactions between adjacent chains in order to implement the two qubit CPHASE gate. Inspired by the role of gauge particles (photons) in quantum electrodynamics which locally interact with electrons to mediate long range interaction between them, we have shown that by adding extra ancillary chains, one can indeed construct such quantum hardware entirely with local interaction between spin chains. In view of the emerging experimental attempts to realize such schemes for quantum state transfer \cite{pst5}, we hope that our bringing of the scheme of \cite{Shor et al.} closer to experimental and practical feasibility  will pave the way for making static and time independent quantum circuits in the future and quantum chips in the long run.  \\
		
		\textbf{Acknowledgements}
We would like to thank A. Bayat for very helpful comments and discussions and for bringing several useful references to our attention.

	\section{Appendix A: The spectrum of the ancillary chain}
	In this appendix we will determine the full spectrum of the Hamiltonian  \eqref{Hanc}. We will need this spectrum when we discuss the errors caused by replacing the full dynamics with an effective one in the ground space of the ancillary chain. The method of obtaining the spectrum is standard \cite{jordan wigner} and is based on the Jordan-Wigner transformation \cite{jordan wigner}:
	\ba
	\phi_j&=&\underbrace{\z \otimes ... \otimes \z \otimes}_\text{$j-1$ times} \frac{\x+i\y}{2} \otimes \mathbb{I} \otimes ...  \otimes \mathbb{I},\cr
	\phi_j^{\dagger} &=&\underbrace{\z \otimes ... \otimes \z \otimes}_\text{$j-1$ times} \frac{\x-i\y}{2} \otimes \mathbb{I} \otimes ...  \otimes \mathbb{I},
	\ea
	where $\phi_j$ and $\phi_j^{\dagger}$ annihilate and create Fermions at site $j$ respectively and obey the anti-commutation relations:
	\begin{equation}
	\{\phi_j,\phi_k^\dagger\}=\delta_{j,k}, \quad \{\phi_j,\phi_k\}=0, \quad \{\phi_j^\dagger,\phi_k^\dagger\}=0.
	\end{equation}
	With this transformation, the model turns into a free Fermion model with Hamiltonian 
	\begin{equation}
	H^\text{anc}=\frac{1}{4m} \sum_{j=0}^{N-1} \phi_j^\dagger\left(2\phi_j-\phi_{j-1}-\phi_{j+1}\right)+\frac{\hat{P}+1}{4m}\left( \phi_{N-1}^\dagger\phi_0 + \phi_0^\dagger\phi_{N-1} \right)
	\end{equation}
	where,
	\begin{equation}
	\hat{P} := \prod_{j=0}^{N-1} \z_j
	\end{equation}
	Note that from \eqref{Hanc},  $[H^\text{anc} , \hat{P}]=0$, Hence, the Hilbert space is divided into two sectors with $\hat{P}=\pm 1$, where $+1$ specifies the sector with even number of excitations and $-1$ with odd number of particles, and we can write:
	\begin{equation}
	H^\text{anc} = \frac{1+\hat{P}}{2} H^+ + \frac{1-\hat{P}}{2} H^-.
	\end{equation}
	In each sector the Fermionic operators can be decoupled by using a Fourier transformation to the normal modes. The difference lies in the boundary conditions for the modes. In the $\hat{P}=+1$ and th e$\hat{P}=-1$ modes, we use respectively periodic and anti-periodic boundary conditions to find 	
	
	In the  $\hat{P}=-1$, sector we find 
	\begin{equation}
	H^{\pm} = \sum_{p=0}^{N-1} {\omega_p}^\pm  {{a_p}^\pm}^\dagger {a_p}^{\pm}
	\end{equation} 
	where,
	\begin{equation}\label{energy2}
	{a^+_p}^\dagger=\frac{1}{\sqrt{N}} \sum_{j=0}^{N-1} e^{\frac{2\pi i}{N}(p+\frac{1}{2})j} \phi_j^\dagger, \qquad \omega^+_p=\frac{1}{m}\sin^2(\frac{\pi}{N}(p+\frac{1}{2})),
	\end{equation} 	
	and
	
	\begin{equation}\label{energy1}
	{a^-_p}^\dagger=\frac{1}{\sqrt{N}} \sum_{j=0}^{N-1} e^{\frac{2\pi ipj}{N}} \phi_x^\dagger, \qquad \omega^-_p=\frac{1}{m}\sin^2(\frac{ \pi p}{N}).
	\end{equation}
	
	In both sectors the energy eigenstates are formed by successive operation of the respective creation operators on the vacuum, where we have suppressed the superscript $\pm$ for simplicity. 
	\begin{equation}\label{spectrum}
	\ket{\tilde{p}_1,\tilde{p}_2,...,\tilde{p}_i}=a_{p_1}^\dagger a_{p_2}^\dagger ... a_{p_i}^\dagger \ket{\Omega}, \qquad \text{with energy: } E = \omega_{p_1} + ... + \omega_{p_i}.
	\end{equation}

	\section{Appendix B: The error analysis}
	
	Putting aside errors which result from imperfections and inhomogeneity in the couplings, the imperfect shape of wave packets and similar source of errors which may result in practice, we are faced with at least three sources of theoretical errors. These  should be properly bounded in order for a quantum computation scheme to work properly. In this appendix we briefly discuss these bounds. \\
	
	 The first type of error and the easiest ones to be dealt with, results from contamination of single particle states with higher particle ones.  If such errors occur, i.e. if a 1-particle sector $|\a\ra$ is contaminated by a 2-particle state $|\beta\ra$ in the form 
$
	\ket{\psi} = \ket{\alpha} + \epsilon\ket{\beta},
$
where, $\ket{\beta}$ is in the 2-particle sector or higher, due  to conservation of particle number, the error $\epsilon$ acquires only a time-varying phase and in the course of time, its magnitude does not increase. At the end of the circuit, where only single particle measurements are performed, such erroneous states are projected out.\\

The other source of error is related with the degree of localization of the packets. This is common to our scheme and that of \cite{Shor et al.} and we suffice to quote from latter reference that for 
 performing $g$ gates on $M$ qubits if  the size of the chains are chosen to  be $N=\Omega\left(M^{3+\delta}g^{3+\delta}\right)$, then the error of the computation will scale as  $O\left(\frac{1}{(Mg)^{\delta/3}}\right)$, for any $\delta>0$. Therefore, this is not a threat to scalability of the scheme. \\
 
 Finally, we come to the third problem which is specific to our way of implementing the CPHASE gate by local interactions with an ancillary chain, where we require that the couplings in the ancillary rail be strong enough. We will show in the rest of this appendix that if the XY couplings of the ancillary rails are larger than $2N^2$ the error that this would add is $O\left(\frac{1}{N^{1/6}}\right)$. Therefore, in view of the polynomial overhead in the size of chains $N$ for bounding all types of errors, the protocol is scalable. The details of this analysis will follow.\\
	
In what follows we will show that if we choose the couplings $\frac{1}{m}$ such that 	
	 $\frac{1}{m} > 2N^2$, then by substituting the Hamiltonian with the effective Hamiltonian, $H_\text{eff}=(\mathbb{I}\otimes P_0)H(\mathbb{I}\otimes P_0)$, the magnitude of error, per each entangling gate, is $O\left(\frac{1}{N^{1/6}}\right)$.
	Specifically we will show that, for each entangling gate block, and a wave packet $\ket{I}$ initialized at the beginning of this gate block, we have:
	\begin{equation}\label{main result}
	\norm{e^{-iH\Delta t}\ket{I}-e^{-iH_\text{eff}\Delta t}\ket{I}} = O\left(\frac{1}{N^{1/6}}\right)
	\end{equation}
	where $\Delta t$ is the time interval that takes for the wave packet to translate through this gate block. We will prove this result by using time-independent perturbation theory. We split the total Hamiltonian into two parts, the base Hamiltonian $\tilde{H}_0=H_\text{eff}$ and the perturbation potential $V'$, which will be defined in the sequel. Then in theorem \ref{gap} we show that if we choose $\frac{1}{m} > 2N^2$, then there is an energy gap, greater than one, between the ground and the excited states of the ancillary rail. Next in theorem \ref{higher order} and \ref{bound} we derive some useful bound that along with theorem \ref{gap} will be used in corollary \ref{error} to derive some bounds on the transition amplitudes between the subspace that is spanned by $\ket{\Omega}$ and $\ket{\Psi}$, and the other eigenstates of $\tilde{H}_0$. Specifically, in corollary \ref{error} by using perturbation theory, we show that if we have an eigenstate of $\tilde{H}_0$, where the ancillary rail is in its ground states, then it is approximately equal to its perturbed ket up to an error $O\left(\frac{1}{N^{1/6}}\right)$, and also its energy will be perturbed with a correction of order $O\left(\frac{1}{N^{5/6}}\right)$. Afterwards, by using corollary \ref{error} we can show that equation \eqref{main result} holds for eigenstates of $\tilde{H}_0$. Given the fact that the wave packet $\ket{I}$ is a linear combination of some eigenstates of $\tilde{H}_0$, we have to show that the same inequality that holds for these eigenstates is true for the wave packet $\ket{I}$ as well. Hence, in theorem \ref{conservation}, we use the translational symmetry of the system to prove an equation for conservation of momentum, and we use it afterwards to prove the main result of this appendix, the equation \eqref{main result}. Also notice that the error analysis in this appendix only is done for the case where the logical state of the two qubits is $\ket{{\bf 1}{\bf 0}}$, i.e. $\ket{I}=\ket{G}^1\otimes\ket{\Omega}^{1'}\otimes\ket{\Omega}^\text{anc}$, the error analysis for the other cases can be done similarly.\\
	
	We provide the error analysis for the case of a $V_X$ gate block, the case of a $V_Y$ gate block is quite similar. First, we split the total Hamiltonian $H=H_\text{free} + V^{1,\text{anc}}_\x+V^{1',anc}_\x$ into two parts, the base Hamiltonian $\tilde{H}_0:=H_\text{free}+V_\text{eff}$, and the perturbation potential $V':=V^{1,\text{anc}}_\x+V^{1',anc}_\x-V_\text{eff}$:
	\begin{equation}\label{H0}
	H = \left(H_\text{free}+V_\text{eff}\right )+ \left( V^{1,\text{anc}}_\x+V^{1',anc}_\x-V_\text{eff}\right) =\tilde{H}_0+V'
	\end{equation}
	where
	\begin{equation}\label{potential}
	H_\text{free} = H^1+H^{1'}+H^\text{anc}, \qquad
	V_\text{eff} =  V^{1,\text{anc}}_\text{x, eff}+V^{1',anc}_\text{x, eff} = e\frac{1}{\sqrt{N}} \left[\hat{N}^\text{1}-\hat{N}^{1'}\right] \otimes \hat{X}^\text{anc},
	\end{equation}
	and
	\begin{equation}\label{V'}
	V' = e\sum_{x}^{} \left[ {\bf n}_x^1 - {\bf n}_x^{1'} \right] \otimes \left[ \x^\text{anc}_x - \frac{\hat{X}^\text{anc}}{\sqrt{N}} \right].
	\end{equation}
	Notice that the base Hamiltonian $\tilde{H}_0$ here, is In fact, equal to the effective Hamiltonian $H_\text{eff}=(\mathbb{I}\otimes P_0)H(\mathbb{I}\otimes P_0)$ defined in section \ref{model}. Now we have to identify the eigenstates of $\tilde{H}_0$ and then obtain the perturbed eigenstates in terms of the perturbation potential $V'$. Given the fact that $\hat{X}^\text{anc}$ can be diagonalizes as:
	\begin{equation}
	\hat{X}^\text{anc} = \ket{+}\bra{+}-\ket{-}\bra{-}, \qquad \text{where:}
	\quad \ket{\pm}=\frac{\ket{\Omega}\pm\ket{\Psi}}{\sqrt{2}},
	\end{equation}
	and the fact that $\ket{\pm}$ are eigenstates of $H^\text{anc}$, if we change the eigenstates $\ket{\Omega}$ and $\ket{\Psi}$ to states $\ket{\pm}$, we can diagonalize both $V_\text{eff}$ and $H_\text{free}$ simultaneously. So, $V_\text{eff}$ commutes with $H_\text{free}$ and in the sector where $\hat{N}^1=1$ and $\hat{N}^{1'}=0$, the eigenstates of $\tilde{H}_0$ can be written as
	\begin{equation}\label{eigenstates}
	\ket{\tilde{p}, \Omega, \alpha} = \ket{\tilde{p}}^1 \otimes \ket{\Omega}^{1'} \otimes \ket{\alpha}^\text{anc}, \qquad \text{with energy:} \quad 2\cos(\frac{2\pi p}{N})+E_\alpha^\text{anc},
	\end{equation}
	where $\ket{\alpha}^\text{anc}$ is an eigenstate of $H^\text{anc}$ with energy $E_\alpha^\text{anc}$ introduced in equation \eqref{spectrum}, except for the eigenstates $\ket{\Omega}$ and $\ket{\Psi}$, which have changed to states $\ket{\pm}$ with energy $E_\pm=\pm\frac{e}{\sqrt{N}}$. From hereafter we assume that $\hat{N}^1=1$ and $\hat{N}^{1'}=0$, and we just do the error analysis for this case, error analysis in other cases can be done similarly. Before going through the perturbation theory calculations, we have to derive some useful bounds.
	\begin{theorem}\label{gap}
		(Energy Gap) Given $\frac{1}{m}>2N^2$ and $\ket{n_0}=\ket{\tilde{p}, \Omega, \pm}$ defined in equation \eqref{eigenstates}, an eigenstate of $\tilde{H}_0$ with energy $E_n^0$. Then for every other eigenstates $\ket{k_0}=\ket{\tilde{q},\Omega,\alpha}$ of $\tilde{H}_0$ with energy $E_k^0$ such that $\ket{\alpha}^\text{anc}\ne\ket{\pm}$, we have:
		\begin{equation}
		\abs{E_n^0 - E_k^0} > 1
		\end{equation}
	\end{theorem}
	\begin{proof}
		Take $\ket{\alpha}^\text{anc}=\ket{\tilde{p}_1,\tilde{p}_2,...,\tilde{p}_i}^\text{anc}$ defined in equation \eqref{spectrum}, then from equation \eqref{eigenstates} we have:
		\begin{align}
		\abs{E_n^0 - E_k^0} &= \abs{2\cos(\frac{2\pi p}{N})\pm\frac{e}{\sqrt{N}} - 2\cos(\frac{2\pi q}{N})-E_\alpha^\text{anc}}
		\ge E_\alpha^\text{anc} - 2\abs{\cos(\frac{2\pi p}{N})} - 2\abs{ \cos(\frac{2\pi q}{N})} - \abs{\frac{e}{\sqrt{N}}} \cr
		&\ge E_\alpha^\text{anc} - 4 - \frac{\abs{e}}{\sqrt{N}}
		\end{align}
		Now notice that:
		\begin{equation}
		\min_{\alpha}\{E_\alpha^\text{anc}\} = \min_{} \left\{\frac{\sin^2\left(\frac{\pi}{N}\right)}{m},\frac{\sin^2\left(\frac{\pi}{N}(\frac{1}{2})\right)}{m} + \frac{\sin^2\left(\frac{\pi}{N}(N-\frac{1}{2})\right)}{m}\right\} = 2\frac{\sin^2(\frac{\pi}{2N})}{m}
		\end{equation}
		hence,
		\begin{equation}
		\abs{E_n^0 - E_k^0}
		\ge 2\frac{\sin^2(\frac{\pi}{2N})}{m} - 4 - \frac{\abs{e}}{\sqrt{N}}
		\approx \frac{1}{2m}\left(\frac{\pi}{N}\right)^2 - 4 > 1
		\end{equation}
		where we have used the assumption that $\frac{1}{m} > 2N^2$. \\
	\end{proof}
	Before going to the next theorem, as we mentioned at the end of section \ref{model}, we have $\frac{e}{\sqrt{N}} = \Theta\left(\frac{1}{\text{gate length}}\right) = \Theta\left(\frac{1}{N^{2/3}}\right)$, Hence, $e=O\left(\frac{1}{N^{1/6}}\right)$ is a vanishingly small number.
	\begin{theorem}\label{higher order}
		For $V'$ the perturbation potential defined in equation \eqref{V'}, we have:
		\begin{equation}
		\sum_{k_1,...,k_{t-1}}^{} \abs{V'_{k_0k_1}V'_{k_1k_2}...V'_{k_{t-1}k_t}} = O\left(e^t\right)
		\end{equation}
		where $V'_{ab}=\bra{a}V'\ket{b}$, and each summation is over some eigenstates of $\tilde{H}_0$ which lie in the section where $\hat{N^1}=1$ and $\hat{N}^{1'}=0$.
	\end{theorem}
	\begin{proof}
		First, we have to prove an useful lemma:
		\begin{lemma}
			Let $\ket{\Psi_1}$ and $\ket{\Psi_2}$ be two arbitrary vectors and $\{\ket{k}\}$ be an orthonormal set of vectors in a Hilbert space. Then there exist an operator $\mathcal{O}$, with $\norm{\mathcal{O}} = 1$, such that:
			\begin{equation}
			\sum_{k}^{} \abs{\braket{\Psi_1}{k}\braket{k}{\Psi_2}} = \bra{\Psi_1}\mathcal{O}\ket{\Psi_2}
			\end{equation}
			where the standard operator norm (largest eigenvalue) is used.
		\end{lemma}
		\begin{proof}
			We can write:
			\begin{equation*}
			\sum_{k}^{} \abs{\braket{\Psi_1}{k}\braket{k}{\Psi_2}}
			= \sum_{k}^{} e^{i\theta_k}\braket{\Psi_1}{k}\braket{k}{\Psi_2}
			=  \bra{\Psi_1}\mathcal{O}\ket{\Psi_2}
			\end{equation*}
			where $\mathcal{O}=\sum_{k}^{} e^{i\theta_k} \ket{k}\bra{k}$, and since $\ket{k}$-s are orthonormal, we have $\norm{\mathcal{O}}=1$.
		\end{proof}
		Now by applying this lemma we can write:
		\begin{align}
		\sum_{k_1,...,k_{t-1}}^{} \abs{V'_{k_0k_1}V'_{k_1k_2}...V'_{k_{t-1}k_t}}
		&= \sum_{k_2,...,k_{t-1}}^{} \bra{k_0}V'\mathcal{O}_1V'\ket{k_2}\abs{V'_{k_2k_3}...V'_{k_{t-1}k_t}} \cr
		&= \sum_{k_3,...,k_{t-1}}^{} \bra{k_0}V'\mathcal{O}_1V'\mathcal{O}_2V'\ket{k_3}\abs{V'_{k_3k_4}...V'_{k_{t-1}k_t}}\cr
		&= \sum_{k_4,...,k_{t-1}}^{} \bra{k_0}V'\mathcal{O}_1V'\mathcal{O}_2V'\mathcal{O}_3V'\ket{k_4}\abs{V'_{k_4k_5}...V'_{k_{t-1}k_t}}
		\end{align}
		So if we continue this until the summation vanishes, we will have:		
		\begin{equation}
		\sum_{k_1,...,k_{t-1}}^{} \abs{V'_{k_0k_1}V'_{k_1k_2}...V'_{k_{t-1}k_t}} = \bra{k_0}V'\mathcal{O}_1V'\mathcal{O}_2 \cdots V'\mathcal{O}_{t-1}V'\ket{k_t} \le \norm{V'\mathcal{O}_1V' \cdots \mathcal{O}_{t-1}V'} \le \norm{V'}^t
		\end{equation}
		where $\norm{\mathcal{O}_i}=1$, and we have used the fact that for every two operators we have $\norm{AB}\le\norm{A}\norm{B}$. To finish the proof, it is enough to show that $\norm{V'} = O(e)$. Notice that here we are restricted in the subspace where $\hat{N}^1=1$, which means the operators ${\bf n}_x^1$ are a set of orthogonal projective operators, Therefore, $V'$ is block diagonal and we have:
		\begin{equation}
		\norm{V'} = \norm{e\sum_{x}^{} {\bf n}^1_x \otimes \left( \x^\text{anc}_x - \frac{\hat{X}^\text{anc}}{\sqrt{N}} \right)} = \abs{e} \max_{x}{\norm{\x^\text{anc}_x - \frac{\hat{X}^\text{anc}}{\sqrt{N}}}}
		\le \abs{e}\max_{x}{\norm{\x^\text{anc}_x}} + \abs{e} \norm{\frac{\hat{X}^\text{anc}}{\sqrt{N}}} = O(e)
		\end{equation}
		since $\norm{\x^\text{anc}_x}=\norm{\hat{X}^\text{anc}}=1$.
	\end{proof}
	\begin{theorem}\label{bound}
		If we have $\frac{1}{m} > 2N^2$, then for $\ket{n_0}=\ket{\tilde{p}, \Omega, \pm}$, an eigenstate of $\tilde{H}_0$ with energy $E_n^0$ defined in equation \eqref{eigenstates}, we have:
		\begin{equation}{\label{bound_eq}}
		\sum_{k\ne n}^{} \abs{\frac{\bra{n_0}V'\ket{k_0}}{E_n^0-E_k^0}} = O\left(\frac{1}{N^{2/3}}\right),
		\end{equation}
		where the summation is over all the eigenstates of $\tilde{H}_0$.
	\end{theorem}
	\begin{proof}
		Without loss of generality assume that $\ket{n_0}=\ket{\tilde{p}, \Omega, +}$. Take $\ket{k_0}=\ket{\tilde{q},\Omega,\alpha}$, then we have:
		\begin{equation}
		\bra{k_0}V'\ket{n_0} = \sum_{x}^{} e\bra{\tilde{q}}{\bf n}^1_{x}\ket{\tilde{p}}
		\bra{\alpha} (\x^\text{anc}_{x} - \frac{\hat{X}^\text{anc}}{\sqrt{N}}) \ket{+}
		\end{equation}
		In cases where $\ket{\alpha}^\text{anc}=\ket{\pm}$, clearly we have $\bra{k_0}V'\ket{n_0}=0$, since $P_0=\ket{+}\bra{+}+\ket{-}\bra{-}$ and consequently $\bra{k_0}V'\ket{n_0}=\bra{k_0}V\ket{n_0}-\bra{k_0}P_0VP_0\ket{n_0}=0$, where $V=V^{1,\text{anc}}_\x+V^{1',\text{anc}}_\x$. Therefore, since $\hat{X}^\text{anc}\ket{\pm}=\pm\ket{\pm}$, we have:
		\begin{equation}
		\bra{k_0}V'\ket{n_0} = \sum_{x}^{} e\bra{\tilde{q}}{\bf n}^1_{x}\ket{\tilde{p}}
		\bra{\alpha}\x^\text{anc}_{x}\ket{+}
		\end{equation}	
		Also notice that $\bra{k_0}V'\ket{n_0}$ vanishes, unless we have either $\hat{N}^1\ket{\alpha}^\text{anc}=\ket{\alpha}^\text{anc}$ or $\hat{N}^1\ket{\alpha}^\text{anc}=2\ket{\alpha}^\text{anc}$, since both $\x^\text{anc}_{x}$ and $\hat{X}^\text{anc}$ either create or annihilate a particle. So, According to eqs. \eqref{energy2}, \eqref{energy1}, and \eqref{spectrum}, $\ket{\alpha}^\text{anc}$ has either of the following forms:
		\begin{enumerate}
			\item $\ket{\alpha}^\text{anc} = \ket{\tilde{\alpha}}^\text{anc} =  \frac{1}{\sqrt{N}} \sum_{x}^{} e^{\frac{2\pi i}{N}\alpha x} \ket{x}$ \label{s1}
			\item $\ket{\alpha}^\text{anc} = \ket{\tilde{n_1},\tilde{n_2}}^\text{anc} = \frac{1}{N} \sum_{x_1,x_2}^{} \epsilon(x_1,x_2) e^{\frac{2\pi i}{N}(\alpha_1x_1 + \alpha_2x_2)} \ket{x_1,x_2}$ \label{s2}
		\end{enumerate}
		where $\alpha$ is integer, while $\alpha_1$ and $\alpha_2$ are half integers. Also $\epsilon(x_1,x_2) = \left\{
		\begin{array}{rl}
		1 & \text{if } x_1 < x_2\\
		-1 & \text{if } x_1 > x_2\\
		0 & \text{if } x_1 = x_2
		\end{array} \right.$. \\
		In case \ref{s1}, since $\bra{\alpha}{\x^\text{anc}_{x}}\ket{\Psi}=0$, we have:
		\begin{align}
		\bra{k_0}V'\ket{n_0} &= \sum_{x}^{}e \bra{\tilde{q}}{\bf n}^1_{x}\ket{\tilde{p}}
		\frac{\bra{\alpha}{\x^\text{anc}_{x}}\ket{\Omega}}{\sqrt{2}} 
		= \sum_{x}^{} \frac{e}{N} e^{\frac{2\pi i}{N}x(p-q)} \sum_{y}^{} \frac{1}{\sqrt{2N}} e^{-\frac{2\pi i}{N}\alpha y}\braket{y}{x} \cr
		&= \frac{e}{\sqrt{2N^3}} \sum_{x}^{} e^{\frac{2\pi i}{N}x(p-q-\alpha)}
		= \frac{e}{\sqrt{2N}}\delta(\alpha+q-p)
		\end{align}
		and in case \ref{s2}, again because $\bra{\alpha}{\x^\text{anc}_{x}}\ket{\Omega}=0$, we have:
		\begin{align}
		\bra{k_0}V'\ket{n_0} &= \sum_{x}^{}e \bra{\tilde{q}}{\bf n}^1_{x}\ket{\tilde{p}}
		\frac{\bra{\alpha}{\x^\text{anc}_{x}}\ket{\Psi}}{\sqrt{2}} \cr
		&= \sum_{x}^{} \frac{e}{N} e^{\frac{2\pi i}{N}x(p-q)} \sum_{x_1,x_2,y}^{} \frac{1}{\sqrt{2N^3}} \epsilon(x_1,x_2) e^{-\frac{2\pi i}{N}(\alpha_1x_1 + \alpha_2x_2)}\bra{x_1,x_2}{X^\text{anc}_{x}}\ket{y} \cr
		&= \frac{e}{\sqrt{2N^5}} \sum_{x_1,x_2,y, x}^{} \epsilon(x_1,x_2) e^{-\frac{2\pi i}{N}(\alpha_1x_1 + \alpha_2x_2+(q-p)x)}(\delta_{x_1,x}\delta_{x_2,y}+\delta_{x_1,y}\delta_{x_2,x}) \cr
		&= \frac{e}{\sqrt{2N^5}} \sum_{x_1,x_2}^{} \epsilon(x_1,x_2) \left[ e^{-\frac{2\pi i}{N}((\alpha_1+q-p)x_1 + \alpha_2x_2)} + e^{-\frac{2\pi i}{N}(\alpha_1x_1 + (\alpha_2+q-p)x_2)} \right]
		\end{align}	
		So, we need to compute:
		\begin{align}
		\sum_{x_1,x_2}^{} \epsilon(x_1,x_2) e^{-\frac{2\pi i}{N}(q_1x_1 + q_2x_2)} &= \sum_{x_1<x_2}^{} e^{-\frac{2\pi i}{N}(q_1x_1 + q_2x_2)} - e^{-\frac{2\pi i}{N}(q_1x_2 + q_2x_1)} \cr
		&= \sum_{x_2}^{} \frac{\omega^{q_1x_2}-1}{\omega^{q_1}-1}\omega^{q_2x_2} - \frac{\omega^{q_2x_2}-1}{\omega^{q_2}-1}\omega^{q_1x_2} \cr
		&= \left( \frac{N-1}{\omega^{q_1}-1} - \frac{N-1}{\omega^{q_2}-1}\right) \delta(q_1+q_2) \cr
		&= i (N-1) \cot(\frac{q_2\pi}{N}) \delta(q_1+q_2)
		\end{align}
		where $\omega=e^{-\frac{2\pi i}{N}}$, and $q_i = q'_i + \frac{1}{2}$ for some integer $q'_i$. Hence:
		\begin{equation}
		\abs{\bra{k_0}V'\ket{n_0}} = \frac{\abs{e}(N-1)}{\sqrt{2N^5}} \abs{\cot(\frac{\alpha_2\pi}{N})-\cot(\frac{\alpha_1\pi}{N})} \delta(\alpha_1+\alpha_2+q-p) \le \frac{\abs{e}}{\sqrt{N}}\delta(\alpha_1+\alpha_2+q-p) 
		\end{equation}
		We have shown that always $\abs{V'_{kn} }\le \frac{\abs{e}}{\sqrt{N}}$, and also we have proven that $V'_{kn}=0$, unless we have $\alpha+q=p$ or $\alpha_1+\alpha_2+q=p$, which is simply an equation for conservation of momentum. Now given the fact that for each $\ket{\alpha}^\text{anc}$ there exits at most one $\ket{k_0}=\ket{\tilde{q},\Omega,\alpha}$ such that $V'_{kn}\ne0$, we can write:
		\begin{equation}
		\sum_{k}^{}\abs{\frac{\bra{n_0}V'\ket{k_0}}{E_n^0-E_k^0}} < \frac{\abs{e}}{\sqrt{N}} \left(\sum_{\alpha=1}^{N-1} \frac{1}{\frac{1}{m}\sin^2(\frac{\pi \alpha}{N}) - 4}
		+ \sum_{\alpha_1\ne \alpha_2}^{} \frac{1}{\frac{1}{m}\sin^2(\frac{\pi \alpha_1}{N}) + \frac{1}{m}\sin^2(\frac{\pi \alpha_2}{N}) - 4}\right)
		\end{equation}
		where in the second summation, $\alpha_i = n_i + \frac{1}{2}$ and $n_i=1,2,...,N$. Now since $\frac{1}{m}>2N^2$, we have $\frac{1}{m}\sin^2(\frac{\pi \alpha}{N}) - 4 > \frac{1}{2m}\sin^2(\frac{\pi \alpha}{N})$, and consequently:
		\begin{equation}
			\sum_{\alpha=1}^{N-1} \frac{1}{\frac{1}{m}\sin^2(\frac{\pi \alpha}{N}) - 4}
			< \sum_{\alpha=1}^{N-1} \frac{1}{\frac{1}{2m}\sin^2(\frac{\pi \alpha}{N})}
		\end{equation}
		Now according to a well-known theorem, for every continuous real valued function $f$ that does not have any local maximum in the interval $(1,n)$, we have:
		\begin{equation}
		f(1)+f(2)+\cdots+f(n) \le \int_{1}^{n} f(x) \,\mathrm{d}x + f(1) + f(n),
		\end{equation}
		hence,
		\begin{equation}
		\sum_{\alpha=1}^{N-1} \frac{1}{\frac{1}{m}\sin^2(\frac{\pi \alpha}{N}) - 4}
		< \int\limits_{1}^{N-1} \frac{2m}{\sin^2(\frac{\pi x}{N})} \,\mathrm{d}x + 2\frac{2m}{\sin^2(\frac{\pi}{N})} 
		\approx \left.{\frac{2mN}{\pi}\cot(x)}\right\rvert_{\frac{\pi}{N}}^{\pi-\frac{\pi}{N}}+\frac{2}{\pi^2} = O(1).
		\end{equation}
		Similarly, for the second term we have:
		\begin{align}
		&\sum_{\alpha_1\ne \alpha_2}^{} \frac{1}{\frac{1}{m}\sin^2(\frac{\pi \alpha_1}{N}) + \frac{1}{m}\sin^2(\frac{\pi \alpha_2}{N}) - 4}
		< \sum_{\alpha_1\ne \alpha_2}^{} \frac{2m}{\sin^2(\frac{\pi \alpha_1}{N}) + \sin^2(\frac{\pi \alpha_2}{N})} \cr 
		&< \int\limits_{\frac{1}{2}}^{N-\frac{1}{2}} \sum_{\alpha_2}^{} \frac{2m}{\sin^2(\frac{\pi x}{N})+\sin^2(\frac{\pi \alpha_2}{N})} \,\mathrm{d}x + 2\frac{2m}{\sin^2(\frac{\pi}{2N})+\sin^2(\frac{\pi \alpha_2}{N})} \cr
		&< \sum_{\alpha_2}^{} \left.{\frac{2mN}{\pi}\frac{\tan^{-1}(\frac{\tan(x)}{\sqrt{\frac{a}{a+1}}})}{\sqrt{a}\sqrt{a+1}}}\right\rvert_{\frac{\pi}{2N}}^{\pi-\frac{\pi}{2N}}+\sum_{\alpha_2}^{} \frac{4m}{\sin^2(\frac{\pi \alpha_2}{N})}  \cr
		&\approx \sum_{\alpha_2}^{} \frac{4m}{\sin^2(\frac{\pi \alpha_2}{N})}  + \sum_{\alpha_2}^{} \frac{4m}{\sin^2(\frac{\pi \alpha_2}{N})} 
		= \sum_{\alpha_2}^{} \frac{8m}{\sin^2(\frac{\pi \alpha_2}{N})} = O(1)
		\end{align}
		where in the third line $a=\sin^2(\frac{\pi \alpha_2}{N})$. From the last two inequalities and the fact that $\frac{e}{\sqrt{N}} = \Theta\left(\frac{1}{N^{2/3}}\right)$, we can conclude the proof:
		\begin{equation}
		\sum_{k}^{}\abs{\frac{\bra{n_0}V'\ket{k_0}}{E_n^0-E_k^0}} = \frac{\abs{e}}{\sqrt{N}} \, O(1) = O\left(\frac{1}{N^{2/3}}\right)
		\end{equation}
	\end{proof}
	Now we can use perturbation theory and make use of the previous theorems to obtain our desired results.
	\begin{corollary}\label{error}
		Let $\ket{n_0}=\ket{\tilde{p}, \Omega, \pm}$ be an eigenstate of $\tilde{H}_0$ with energy $E_n^0$ defined in equation \eqref{eigenstates}, and let $V'$ be the perturbation potential defined in \eqref{potential}. Also define $\ket{n}$ to be the normalized perturbed eigenket of $\ket{n_0}$, then we have:
		\begin{equation}
		\sqrt{1-\abs{\braket{n_0}{n}}^2 } = O\left(\frac{1}{N^{1/6}}\right), \qquad \Delta_n = O\left(\frac{1}{N^{5/6}}\right)
		\end{equation}
		where $\Delta_n = E_n - E_n^0$, and $E_n$ is the perturbed energy. 
	\end{corollary}
	\begin{proof}
		According to time-independent perturbation theory we have:
		\begin{align}
		\sqrt{1-\abs{\braket{n_0}{n}}^2 } &=
		\sqrt{\sum_{k \ne n}^{} \abs{\braket{k_0}{n}}^2 } \cr
		&= \sqrt{\sum_{k \ne n}^{} \abs{\frac{V'_{kn}}{E_n^0-E_k^0}+\sum_{l\ne n}^{} \frac{V'_{kl}V'_{ln}}{(E_n^0-E_k^0)(E_n^0-E_l^0)}- \frac{V'_{nn}V'_{kn}}{\left(E_n^0-E_k^0\right)^2}+ \cdots }^2} \cr
		&\le \sqrt{\sum_{k \ne n}^{} \frac{\abs{V'_{nk}V'_{kn}}}{\Delta_{nk}^2}
			+ 2\sum_{k, l\ne n}^{} \frac{\abs{V'_{nk}V'_{kl}V'_{ln}}}{\Delta_{nk}^2\Delta_{nl}}
			+ \sum_{k, l, m\ne n}^{} \frac{\abs{V'_{nl}V'_{lk}V'_{km}V'_{mn}}}{\Delta_{nk}^2\Delta_{nl}\Delta_{nm}}
			+ \cdots}
		\end{align}
		where $\Delta_{nk}=\abs{E_n^0-E_k^0}$, and we have used the fact that $V'_{nn}=0$. Applying theorem \ref{gap} and \ref{higher order} gives:
		\begin{equation}
		\sqrt{1-\abs{\braket{n_0}{n}}^2 } \le \sqrt{\frac{1}{\Delta^2}\abs{e}^2 + 
			2\frac{1}{\Delta^3}\abs{e}^3 + 3\frac{1}{\Delta^4}\abs{e}^4 + \cdots }
		= O(e) = O\left(\frac{1}{N^{1/6}}\right)
		\end{equation}
		where $\displaystyle{\Delta = \min_{k \ne n}{\abs{E_n^0-E_k^0}}}\ge1$.
		Now notice that:
		\begin{equation*}
		\bra{n_0}\left(\tilde{H}_0-E_n^0-\Delta_n\right)\ket{n}=-\bra{n_0}V'\ket{n} \quad \Rightarrow \quad \Delta_n=\frac{\bra{n_0}V'\ket{n}}{\braket{n_0}{n}}
		= \sum_{k \ne n}^{} \frac{\bra{n_0}V'\ket{k_0}\braket{k_0}{n}}{\braket{n_0}{n}}
		\end{equation*}
		therefore,
		\begin{align}
		\abs{\Delta_n} &\le \sum_{k\ne n}^{} \abs{\frac{V'_{nk}\braket{k_0}{n}}{\braket{n_0}{n}}} \approx \sum_{k \ne n}^{} \abs{V'_{nk}\braket{k_0}{n}} \cr
		&= \abs{\sum_{k \ne n}^{} \frac{V'_{nk}V'_{kn}}{\Delta_{nk}}+\sum_{k, l\ne n}^{} \frac{V'_{nk}V'_{kl}V'_{ln}}{\Delta_{nk}\Delta_{nl}}
		+ \sum_{k, l, m\ne n} \frac{V'_{nk}V'_{kl}V'_{lm}V'_{mn}}{\Delta_{nk}\Delta_{nl}\Delta_{nm}}+\cdots} \cr
		&\le \sum_{k \ne n}^{} \abs{\frac{V'_{nk}V'_{kn}}{\Delta_{nk}}}+\sum_{k, l\ne n}^{} \abs{\frac{V'_{nk}V'_{kl}V'_{ln}}{\Delta_{nk}\Delta_{nl}}}
		+ \sum_{k, l, m\ne n} \abs{\frac{V'_{nk}V'_{kl}V'_{lm}V'_{mn}}{\Delta_{nk}\Delta_{nl}\Delta_{nm}}}+\cdots \cr
		&\le \sum_{k \ne n}^{} \abs{\frac{V'_{nk}}{\Delta_{nk}}}
		\left(\abs{V'_{nk}}+\sum_{l\ne n}^{} \frac{\abs{V'_{kl}V'_{ln}}}{\Delta}
		+ \sum_{l, m\ne n} \frac{\abs{V'_{kl}V'_{lm}V'_{mn}}}{\Delta^2}+\cdots\right)
		\end{align}
		then applying theorems \ref{gap} and \ref{higher order}, gives:
		\begin{equation}
		\le \sum_{k \ne n}^{} \abs{\frac{V'_{nk}}{\Delta_{nk}}}
		\left(\abs{e} + 
		\abs{e}^2 + \abs{e}^3 + \cdots\right) 
		= O\left(\frac{e^2}{\sqrt{N}}\right) = O\left(\frac{1}{N^{5/6}}\right)
		\end{equation}
		where we have used equation \eqref{bound_eq} from theorem \ref{bound}.
	\end{proof}
	The last theorem that we need, is a theorem for conservation of momentum.
	\begin{theorem}\label{conservation}
		(Conservation of Momentum) if we define the translation operator  $\hat{T}$ as:
		\begin{equation}
		\hat{T}:=T^1 \otimes T^2 \otimes T^\text{anc},
		\end{equation}
		where for each chain, the translation operator $T$ is defined as:
		\begin{equation}
		T\ket{a_0,a_1,...,a_{N-1}}=\ket{a_{N-1},a_0,...,a_{N-2}}, \qquad \text{where:} \quad a_i\in \{0,1\}
		\end{equation}
		then:
		\begin{enumerate}
			\item $[\hat{T},H]=0$, where $H=H_\text{free} + V^{1,\text{anc}}_\x+V^{2,anc}_\x$ is the total Hamiltonian of the system. \label{statement1}
			\item $\hat{T}\ket{\tilde{p},\Omega,+}=e^{-\frac{2\pi i}{N}p}\ket{\tilde{p},\Omega,+}$. \label{statement2}
			\item for every $p\ne q: \quad \bra{\tilde{q},\Omega,+}e^{iHt}\ket{\tilde{p},\Omega,+}=0$. \label{statement3}
		\end{enumerate}	
	\end{theorem}
	\begin{proof}
		\ref{statement1}: We can write the total Hamiltonian as:
		\begin{equation}
		H = \sum_{j=0}^{N-1} {\bf h}_{j,j+1}
		\end{equation}
		where each ${\bf h}_{j,j+1}$ is a sum of some local operators acting on the $j$-th and $(j+1)$-th sites of each chain. Then because of the translational symmetry of the system we have:
		\begin{equation}
		\hat{T}^\dagger \, {\bf h}_{j,j+1} \, T = {\bf h}_{j+1,j+2}
		\end{equation}
		Therefore, we have $\hat{T}^\dagger HT=H$, and since $\hat{T}$ is unitary we can conclude that $[\hat{T},H]=0$. \\
		\ref{statement2}: First note that we have:
		\begin{equation}
		T\ket{\Omega} = T \ket{0,0,...,0} = \ket{\Omega} \qquad \text{and} \qquad
		T^\text{anc}\ket{\Psi} = \frac{1}{\sqrt{N}} \sum_{x}^{} T^\text{anc} \ket{x}
		= \frac{1}{\sqrt{N}} \sum_{x}^{} \ket{x+1} = \ket{\Psi}
		\end{equation}
		Hence, we have $T\ket{+}=\ket{+}$, and therefore:
		\begin{align}
		\hat{T}\ket{\tilde{p},\Omega,+} &= \frac{1}{\sqrt{N}} \sum_{x}^{} e^{\frac{2 \pi i}{N}px} \, \hat{T} \ket{x}\otimes\ket{\Omega,+}
		= \frac{1}{\sqrt{N}} \sum_{x}^{} e^{\frac{2 \pi i}{N}px} \ket{x+1}\otimes\ket{\Omega,+}
		\cr
		&= \frac{1}{\sqrt{N}} \sum_{x}^{} e^{\frac{2 \pi i}{N}p(x-1)} \ket{x}\otimes\ket{\Omega,+} = e^{-\frac{2\pi i}{N}p}\ket{\tilde{p},\Omega,+}
		\end{align} \\
		\ref{statement3}: From \ref{statement1} we have $[\hat{T},e^{iHt}]=0$, therefore:
		\begin{equation}
		\bra{\tilde{q},\Omega,+}\hat{T}e^{iHt}\ket{\tilde{p},\Omega,+}=\bra{\tilde{q},\Omega,+}e^{iHt}\hat{T}\ket{\tilde{p},\Omega,+}
		\end{equation}
		and so from \ref{statement2} we have:
		\begin{equation}
		\left(e^{-\frac{2\pi i}{N}q}-e^{-\frac{2\pi i}{N}p}\right)\bra{\tilde{q},\Omega,+}e^{iHt}\ket{\tilde{p},\Omega,+}=0
		\end{equation}
		therefore, since $p\ne q$ we can conclude that $\bra{\tilde{q},\Omega,+}e^{iHt}\ket{\tilde{p},\Omega,+}=0$.
	\end{proof}
	
	Now we are ready to compute the error that we get by substituting the Hamiltonian by the effective Hamiltonian $\tilde{H}_0$, when performing a controlled phase operation. In our model a controlled phase operation is implemented by a gate block localized somewhere along the chains. The initial state of the system which In fact, is a superposition of Gaussian packets will enter this gate block, evolves with the Hamiltonian for the system $H$, and then exit this gate block. Now suppose that the initial state be $\ket{I}$, then after time $\Delta t$ when the wave packets exit this gate block the state of the system would be $e^{-itH}\ket{I}$. Therefore, in this process the error that we will obtain by the substitution is just $\norm{e^{-iH\Delta t}\ket{I}-e^{-i\tilde{H}_0\Delta t}\ket{I}}$. Hence, we have to show that this number is vanishingly small. Note that here we do not have to be worried about the state of other rails in the system. Because their ring Hamiltonians commute with the perturbation potential $V'$. Therefore, the unitary evolution of the system is separable and they evolve only through their ring Hamiltonians. \\	
	Because the initial state of the ancillary rail is $\ket{\Omega}$, then we can write the overall initial state of the two 1 rails and the ancillary rail as:
	\begin{equation}
	\ket{I} = \frac{\ket{I^+}+\ket{I^-}}{\sqrt{2}}, \qquad \text{where:} \quad \ket{I^\pm}=\ket{G,\Omega,\pm} = \sum_{p}^{} a_p \ket{\tilde{p},\Omega,\pm}
	\end{equation}
	Now take $\ket{n_p^0}=\ket{\tilde{p},\Omega,+}$, and its normalized perturbed ket to be $\ket{n_p}=\alpha\ket{n_p^0}+\beta\ket{\varphi_p}$ such that $\braket{n_p^0}{\varphi_p}=0$. Then we get that:
	\begin{equation}
		\left(e^{-iH\Delta t}-e^{-i\left(E_n^0+\Delta_n\right)\Delta t}\right)\ket{n_p}=0
	\end{equation}
	where $E_n^0+\Delta_n$ is the perturbed energy eigenstate. Then:
	\begin{equation}
	e^{-iH\Delta t}\ket{n_p^0} = 
	e^{-i\left(E_n^0+\Delta_n\right)\Delta t}\ket{n_p^0}-\frac{\beta}{\alpha}\left(e^{-iH\Delta t}-e^{-i\left(E_n^0+\Delta_n\right)\Delta t}\right)\ket{\varphi_p}
	\end{equation} 
	so according to corollary \ref{error}, $\beta = O\left(\frac{1}{N^{1/6}}\right)$ and we get that:
	\begin{align}
	\norm{e^{-iH\Delta t}\ket{n_p^0}-e^{-i\tilde{H}_0\Delta t}\ket{n_p^0}} &= \norm{e^{-iE_n^0\Delta t}\left(e^{-i\Delta_n\Delta t}-1\right)\ket{n_p^0}-\frac{\beta}{\alpha}\left(e^{-iH\Delta t}-e^{-i\left(E_n^0+\Delta_n\right)\Delta t}\right)\ket{\varphi_p}} \cr
	&\le \abs{e^{-i\Delta_n\Delta t}-1} + 2\abs{\frac{\beta}{\alpha}} 
	= O\left(\Delta_n\Delta t\right) + O\left(\frac{1}{N^{1/6}}\right) \cr
	&= O\left(\frac{1}{N^{5/6}}\Delta t+\frac{1}{N^{1/6}}\right) 
	= O\left(\frac{1}{N^{1/6}}\right)
	\end{align}
	where in the last line we have used the fact that $\Delta t=\Theta\left(N^{2/3}\right)$, because $\Delta t$ is twice the length of a gate block, since the group velocity of our wave packets is equal to $2$. Now given \ref{statement3} from theorem \ref{conservation}, since the vectors $\left\{\ket{n_p^0}\right\}$ have different momenta, the vectors $\left(e^{-iH\Delta t}-e^{-i\tilde{H}_0\Delta t}\right)\ket{n_p^0}$ are orthogonal, and consequently:
	\begin{equation}
	\norm{e^{-iH\Delta t}\ket{I^+}-e^{-i\tilde{H}_0\Delta t}\ket{I^+}}^2=
	\sum_{p}^{} {a_p}^2 \norm{e^{-iH\Delta t}\ket{n_p^0}-e^{-i\tilde{H}_0\Delta t}\ket{n_p^0}}^2 = O\left(\frac{1}{N^{2/6}}\right)
	\end{equation}
	since $\sum_{p}^{} {a_p}^2=1$. Similarly, we can obtain the same inequality for $\ket{I^-}$ as well. Finally, we obtain:
	\begin{align}
	\norm{e^{-iH\Delta t}\ket{I}-e^{-i\tilde{H}_0\Delta t}\ket{I}} &\le \frac{\norm{e^{-iH\Delta t}\ket{I^+}-e^{-i\tilde{H}_0\Delta t}\ket{I^+}}+\norm{e^{-iH\Delta t}\ket{I^-}-e^{-i\tilde{H}_0\Delta t}\ket{I^-}}}{\sqrt{2}} \cr
	&= O\left(\frac{1}{N^{1/6}}\right)
	\end{align}
	Hence, if we have $\frac{1}{m} > 2N^2$, then each entangling gate will only add a small error $O\left(\frac{1}{N^{1/6}}\right)$.	
		
\end{document}